%% file: main.tex
\documentclass[lettersize, journal]{IEEEtran}
\usepackage{cite}
\ifCLASSINFOpdf
\fi
\usepackage{algorithm}
\usepackage{algpseudocode}  
\usepackage{amsmath,amsfonts,amsthm}
\usepackage{amssymb}
\usepackage{tikz}
\usepackage{pgfplots}
\usepgfplotslibrary{fillbetween}
\usepackage{float}
\usepackage{algpseudocode}
\usepackage{tikz}
\usepackage{graphicx}
\usetikzlibrary{matrix}
\usepackage{bbm}
\usepackage{graphicx}
\usepackage{subcaption}
\usepackage{siunitx}
\usetikzlibrary{arrows,automata}
\usepackage{multicol}
\usepackage{comment}
\usepackage{epstopdf}
\usepackage{multirow}
\usepackage[font=scriptsize,skip=5pt]{caption}
\usepackage{soul}
\usepackage{todonotes}
\usepackage{color, colortbl}
\usepackage{array}
\usetikzlibrary{patterns}
\usepackage{balance}

\usepackage{enumitem}
\usepackage{footmisc}
\usetikzlibrary{patterns,fillbetween}
\usepackage{wrapfig}

\newtheorem{theorem}{Theorem}

\newtheorem{lemma}[theorem]{Lemma}

\newtheorem{remark}{Remark}

\allowdisplaybreaks

\begin{document}
\title{Version AoI Optimization under Power and General Distortion Constraints in Uplink NOMA}

\author{
Gangadhar Karevvanavar,
Rajshekhar V.~Bhat,
and Nikolaos~Pappas,~\IEEEmembership{Senior Member,~IEEE}

    \thanks{Gangadhar Karevvanavar and Rajshekhar V.~Bhat are with the Indian Institute of Technology Dharwad, Dharwad, India
(e-mail: 212021007@iitdh.ac.in; rajshekhar.bhat@iitdh.ac.in).

Nikolaos~Pappas is with the Dept. of Computer and Information Science,
Linköping University, Linköping, Sweden
(e-mail: nikolaos.pappas@liu.se).}
    \thanks{
        The work of Rajshekhar V. Bhat was supported by the TTDF, DoT, Government of India, under Proposal ID TTDF/6G/492 through TCOE India.} 
\thanks{The work of Nikolaos Pappas was supported by ELLIIT and the European Union (6G-LEADER, 101192080).
    }
}

\maketitle
\pagenumbering{arabic}

\begin{abstract}
The Version Age of Information (VAoI) quantifies information freshness by measuring the number of versions the receiver lags behind. This paper studies VAoI minimization in an $M$-user uplink non-orthogonal multiple access (NOMA) system where users maintain single-packet buffers and transmissions are constrained by average power and information-quality constraints, modeled by a general distortion function. A fundamental trade-off arises: transmitting more bits per update improves information quality but increases power consumption, reducing transmission opportunities and increasing VAoI, while transmitting fewer bits has the opposite effect. We formulate a weighted-sum VAoI minimization problem as a convex optimization problem. However, users’ power allocations are coupled through multiple-access capacity constraints per channel state, leading to exponential complexity. To address this, we develop a VAoI-agnostic stationary randomized policy  that jointly optimizes scheduling, bit allocation, and power control without tracking instantaneous VAoI, and achieves a provable 2-approximation to the globally optimal average VAoI. Leveraging Lagrangian dual decomposition, we derive closed-form expressions for the scheduling probabilities and power allocations, and efficiently determine the optimal successive interference cancellation decoding order, avoiding exhaustive search Numerical results show that NOMA significantly outperforms time-division multiple access (TDMA): at high power budgets, NOMA achieves near-zero VAoI, whereas TDMA saturates at a non-zero value, consistent with the analysis. The proposed general distortion framework accommodates diverse bit-priority structures by assigning unequal importance to different bits within an update.
\end{abstract}


\IEEEpeerreviewmaketitle

\section{Introduction}
Real-time monitoring and control applications, such as vehicular networks and industrial Internet of Things (IoT), require not only timely updates but also sufficient information quality for effective decision-making. In such systems, a central node collects information from multiple distributed sources. Traditional metrics such as throughput and latency do not capture information freshness, which is quantified by the age of information (AoI), defined as the time elapsed since the generation of the freshest received update at the destination~\cite{kaul2012real,pappas2023age,kosta2017age,sun2019age}. While AoI captures timeliness, it treats all successfully received updates equally, irrespective of their completeness or informational significance. Version age of information (VAoI) refines this notion by measuring version staleness rather than time elapsed, i.e., how many versions the receiver lags behind~\cite{yates2021age, Gangadhar_VAoI,baturalp2022version, DelfaniTCOM25, DelfaniCOMML25, DelfaniINFOCOM26}. However, neither AoI nor VAoI accounts for update quality, motivating the incorporation of distortion metrics.

In practice, not all bits within an update carry equal importance. For instance, in sensor networks, initial bits may encode critical metadata (e.g., anomaly flags), while later bits provide finer-grained measurements. Similarly, in video streaming, header information and the essential coarse components of intra-coded (I-) frames are typically prioritized over fine-grained refinement data. To capture such heterogeneous bit priorities, we model information quality using a general distortion metric $\delta(\rho)$ that depends on the number of successfully transmitted bits $\rho \in \{0,1,\ldots,r_{\max}\}$, without restricting $\delta(\cdot)$ to be convex or monotonically decreasing, as is commonly assumed \cite{WiOpt-21,AoI-Dist-Tradeoff4}. This framework captures scenarios where initial bits convey task-critical information, and later bits provide refinement, with $\delta(\cdot)$ assigning higher penalties to incomplete transmission of high-priority bits.

In this work, we aim to minimize the long-term average VAoI in $M$-user uplink multiple access channel (MAC) subject to average power and distortion constraints. A fundamental trade-off arises between timeliness and information quality: transmitting fewer bits per update conserves power and enables more frequent updates, reducing VAoI at the cost of higher distortion, while transmitting more bits improves information quality but consumes more power and reduces transmission opportunities, potentially increasing VAoI. This trade-off becomes more severe in uplink non-orthogonal multiple access (NOMA) systems, where users’ power allocations are coupled through interference, and the decoding strategy, such as the successive interference cancellation (SIC), directly impacts both achievable rates and required transmit powers.

We review related work to position our contributions. The AoI framework, studied in~\cite{kaul2012real,pappas2023age,kosta2017age,sun2019age}, characterized the trade-off between update frequency and information freshness and was later extended to VAoI to capture version-based staleness~\cite{yates2021age}. These concepts have been extended to multiple sources, energy harvesting, and advanced multiple access~\cite{yates2018age, costa2016age}. Minimizing AoI and its variants under practical constraints, such as power, distortion, and energy harvesting, yields challenging joint optimization problems. Approximation algorithms and stationary randomized policies have emerged as effective low-complexity solutions with provable performance guarantees~\cite{kadota2018optimizing, NetworkWithStochasticArrivals, Bhat-JSAC,goal_oriented_srp_2024, kumar2023, wiopt_srp_2024, random_access_age_2020}. In this work, we develop a low-complexity stationary randomized policy and derive a performance bound.

To support simultaneous uplink transmissions in the MAC, we adopt NOMA, which enables users to share resources via superposition coding and receiver-side SIC. Several recent works have investigated NOMA for AoI optimization. In~\cite{zhiguo2023age}, a time-division multiple access (TDMA) network is augmented with two-user cognitive-radio NOMA, where the secondary user is decoded first with a limited rate to ensure primary-user priority, showing that NOMA improves AoI by increasing transmission opportunities and reducing access delay. The uplink coexistence of mission-critical (MC) and enhanced mobile broadband (eMBB) services has been studied by comparing puncturing and NOMA: puncturing yields fresher MC updates, whereas NOMA can increase eMBB rates by up to five times at the cost of some freshness loss~\cite{khodakhah2024balancing,sun2025ageinformationanalysisnomaassisted}. NOMA-assisted grant-free uplink schemes with randomly arriving packets have also been shown to reduce AoI compared to OMA-based approaches~\cite{sun2025ageinformationanalysisnomaassisted} significantly. In contrast to some of these works, which focus on two-user systems, we consider a general $M$-user setting and jointly optimize transmitting bits and transmit powers under average power and distortion constraints, as identified as a future direction in~\cite{zhiguo2023age}.

Several works study AoI optimization with distortion, highlighting the trade-off between timeliness and update quality. AoI minimization under constant and age-dependent distortion constraints shows that longer processing improves quality at the cost of increased AoI~\cite{9524846}. Online scheduling formulations balancing AoI, energy, and quality have been proposed, with greedy algorithms achieving provable competitiveness~\cite{AoI-Dist-Tradeoff4}. Related metrics, such as the age of incorrect estimates and age of incorrect information (AoII), jointly capture freshness and accuracy, leading to threshold-type optimal policies~\cite{9589182,9137714}. Joint sampling and compression strategies have been shown to achieve asymptotically optimal age--distortion trade-offs~\cite{10715699}. In discrete-time settings, optimal age--distortion policies have been characterized via dynamic programming~\cite{9611456}, and deep reinforcement learning has been applied to semantic-empowered NOMA systems for AoII minimization under power and bandwidth constraints~\cite{10336741,Guidan}. Collectively, these works underscore the importance of distortion-aware policies for ensuring timely, high-quality updates.

Our work is most closely related to~\cite{Bhat-JSAC, WiOpt-21, Gangadhar_VAoI}. In~\cite{Bhat-JSAC}, AoI minimization in fading uplink MACs under TDMA and NOMA is studied. While both works consider NOMA-based uplink MACs, our approach differs in key aspects. \emph{First}, we optimize VAoI~\cite{yates2021age,baturalp2022version}, which measures version staleness and increases only upon new arrivals, capturing both timeliness and content staleness and inducing a different Markov structure than AoI. \emph{Second}, we incorporate general distortion metric $\delta(\rho)$ enabling flexible bit-priority modeling (e.g., critical metadata in initial bits~\cite{pappas2023age,gunduz_sc2019} and refinement in later bits), leading to joint optimization of \emph{whether} to transmit and \emph{how many bits} to transmit per user—a trade-off absent in~\cite{Bhat-JSAC}. Compared to~\cite{WiOpt-21}, which studies AoI minimization with power–distortion constraints under TDMA-only MAC, we consider NOMA, introducing $2^M-1$ MAC constraints for all feasible user rate vectors; moreover, while~\cite{WiOpt-21} assumes a fixed distortion function, our general $\delta(\rho)$ captures a broad class of bit-priority structures. Relative to~\cite{Gangadhar_VAoI}, which studies VAoI minimization in a broadcast channel, the uplink MAC considered here exhibits different interference coupling, requiring coordinated power allocation to satisfy MAC constraints across rate vectors, whereas~\cite{Gangadhar_VAoI} assumes base-station-controlled transmissions with fixed bit allocation. Finally, distortion constraints require new analytical tools for jointly optimizing scheduling, bit allocation, and power control.

Several other works address related but distinct problems. Works such as~\cite{AoI-Dist-Tradeoff4,9589182,9137714} incorporate information quality via AoII or estimation-error metrics but do not study the joint optimization of VAoI and distortion under NOMA. Recent studies~\cite{zhiguo2023age,khodakhah2024balancing,sun2025ageinformationanalysisnomaassisted} analyze AoI in NOMA systems but are restricted to two-user settings and do not incorporate distortion constraints. Works on AoI--distortion trade-offs~\cite{9524846,10715699,9611456} focus on single-user or TDMA systems, thereby avoiding the complexity of NOMA-MAC constraints. Although power allocation in MAC is well studied in the capacity literature~\cite{Sumei_Sum_BC_Region}, extending these techniques to VAoI optimization under distortion constraints requires new formulations and methods. The main contributions of this work are as follows:
\begin{itemize}[leftmargin=*]
    \item We formulate VAoI minimization in uplink NOMA systems under  average power and distortion constraints as a convex optimization problem. Unlike prior NOMA-based age optimization works~\cite{Bhat-JSAC,zhiguo2023age}, we explicitly incorporate information-quality constraints that capture the trade-off between update frequency and distortion. A flexible distortion metric captures diverse bit-priority structures, from uniform importance to metadata-critical updates.

    \item We propose a VAoI-agnostic, stationary, randomized policy (VA-SRP) that operates without tracking instantaneous VAoI and achieves a provable 2-approximation to the globally optimal VAoI across all policies. Leveraging the structure of the MAC capacity region, we employ Lagrangian dual
decomposition and derive closed-form expressions for the
scheduling probabilities and power allocations via SIC, efficiently handling
the $2^M - 1$ MAC constraints. The optimal SIC decoding order is obtained via sorting in $\mathcal{O}(M\log M)$ time, avoiding exhaustive $\mathcal{O}(M!)$ search.

    \item Numerical results demonstrate substantial VAoI gains of NOMA over TDMA. At high power budgets, NOMA achieves near-zero VAoI by enabling simultaneous multi-user transmissions with status update delivery probabilities approaching unity. In contrast, TDMA saturates at a non-zero VAoI due to its single-user-per-slot limitation. We compare VA-SRP with heuristic policies, demonstrating performance and computational advantages, and illustrate the impact of distortion structures, including uniform (linear), metadata-critical (step), and refinement-oriented (convex) models.
\end{itemize}

\section{System Model}
We consider a system with \( M \) sources transmitting status updates to a base station (BS) in slotted time. We describe the packet arrival, channel, distortion, power, and VAoI evolution models, and formulate a long-term average VAoI minimization problem subject to average power and distortion constraints.

\subsection{Packet Arrival Model}
Each user $i \in \mathcal{M} \triangleq \{1,2,\ldots,M\}$ receives a new version update packet at the beginning of each slot with probability $\lambda_i$. Each user maintains a \emph{single-packet} queue that stores only the most recent packet, rendering previous packets obsolete. Let $A_i(t)$ be an indicator variable such that $A_i(t)=1$ if the $i^{\mathrm{th}}$ user receives a new packet at time~$t$, and $A_i(t)=0$ otherwise. The sequence $\{A_i(t)\}_{t\ge1}$ is an independent and identically distributed (i.i.d.) Bernoulli process with $\mathbb{P}(A_i(t)=1)=\lambda_i$, for all $i\in\mathcal{M}$ and $t \in \{1, 2, \ldots\}$.

Each update packet contains $r_{\max}$ bits. In slot~$t$, User~$i$ transmits $\rho_i(t)\in\{0,1,\ldots,r_{\max}\}$ bits, where $\rho_i(t)>0$ corresponds to a valid update. The quality of the received update is captured by the distortion metric $\delta(\rho_i(t))$ defined in Section~II-C and depends on the number of successfully transmitted bits. For instance, when initial bits encode critical information, and later bits provide refinement, $\delta(\cdot)$ can reflect this priority structure. If a packet is not transmitted, is considered for transmission in the next slot unless replaced by a new arrival. If $\rho_i(t) > 0$ bits are transmitted in a slot, the remaining $r_{\max} - \rho_i(t)$ bits are discarded at the end of the slot, and the queue is deemed empty. 
We assume unsent bits are dropped rather than buffered to avoid state-space explosion and to reflect grant-free NOMA and ultra-reliable low latency
communication (URLLC) operation, finite-blocklength coding, and AoI/VAoI packet management~\cite{dai2015non,islam2017power,polyanskiy2010channel,kaul2012real,yates2021age}. This yields a tractable VAoI model that depends only on current decisions and admits stationary randomized policies with provable guarantees, while ignoring cross-slot distortion accumulation, which is left for future work.

\subsection{Channel Model}
We model the wireless channel as block fading, remaining constant within a slot and changing independently across slots. Let \( H_i(t) \) denote the channel power gain between User \( i \) and the BS in slot \( t \), where \( \{H_i(t)\}_{t\geq 1} \) are i.i.d. random variables taking values in a finite set \( \mathcal{H}_i \subset \mathbb{R}^+ \). The channel realization \( h_i(t) \) is perfectly estimated at the beginning of each slot using pilot signals. The channel state vector \( \mathbf{h}(t) \triangleq (h_1(t),\ldots,h_M(t)) \) is known at the start of each slot. We denote \( \boldsymbol{\mathcal{H}} \)  as the the set of all channel states.

\subsection{Distortion Model}
Let $Q_i(t)\in\{0,1\}$ denote whether User~$i$ has a pending packet at the beginning of slot~$t$, i.e., $Q_i(t)=1$ if User~$i$ has a packet in the queue, and $Q_i(t)=0$ otherwise. Per-slot distortion is defined as
\[
d_i(t) \;=\;
\begin{cases}
\delta(\rho_i(t))\, \mathbb{I}\{\rho_i(t) > 0\}, & \text{if } Q_i(t) = 1,\\[6pt]
0, & \text{if } Q_i(t) = 0,
\end{cases}
\]
where $\delta:[0,r_{\max}]\to\mathbb{R}^+$ is a function of transmitted bits $\rho_i(t)$ that captures bit importance, e.g., higher penalties when high-priority bits are missing and diminishing returns for extra bits. The distortion is zero when $\rho_i(t)=0$ or when $Q_i(t)=0$, corresponding to no transmission or empty queue.

An alternative model could incur distortion whenever a packet is pending but not transmitted, but this would cause VAoI and distortion to evolve similarly, coupling timeliness and quality. In our formulation, distortion captures \emph{transmission quality}, while VAoI captures \emph{staleness}, enabling independent control of frequent low-quality versus infrequent high-quality updates. Models where distortion accrues during deferral are better suited to settings with intrinsic information decay and are beyond the scope of this work.

\subsection{Rate-Power Relationship}
In an uplink MAC using the NOMA scheme, users transmit simultaneously to the BS. In slot $t$, User $i\in\mathcal{M}$ transmits $\rho_i(t)\in\{0,1,\ldots,r_{\max}\}$ bits using power $f_i(t)$ under channel power gain $h_i(t)$. All users encode and transmit $(\rho_1(t),\ldots,\rho_M(t))$ bits simultaneously using powers $(f_1(t),\ldots,f_M(t))$; if $\rho_i(t)=0$, then $f_i(t)=0$. The BS receives a superimposed signal comprising
the simultaneous transmissions from all users and decodes all users’ transmissions. For successful decoding under channel state $\mathbf{h}(t)$, the rate–power pairs $(\rho_i(t),f_i(t))_{i\in\mathcal{M}}$ must satisfy
\begin{align}\label{eq:MAC}
    \sum_{i \in \mathcal{S}} \rho_i(t) \leq g\!\left( \sum_{i \in \mathcal{S}} f_i(t) h_i(t) \right), \;\; \forall \mathcal{S} \subseteq \mathcal{M}, \; \forall t.
\end{align}
Here $g(\cdot)$ is a concave, non-decreasing function with $g(0)=0$. The MAC capacity region induces $2^M - 1$ constraints when users transmit simultaneously using the NOMA scheme.

\subsection{Performance Metric and Problem formulation}
Let $z_i(t)$ denote the version index of the packet in User~$i$’s queue at the beginning of slot~$t$, given by $z_i(t)=\sum_{\tau=1}^t A_i(\tau)$, which increments only upon new arrivals. Let $y_i(t)$ denote the version of the most recently received packet from User~$i$ at the BS, which updates only when $\rho_i(t)>0$. The instantaneous VAoI is defined as $\Delta_i(t) = z_i(t) - y_i(t)$, measured at the end of each time slot, and evolves as follows:
\begin{align}\label{eq:age_evolution}
    \Delta_i(t) = \big( \Delta_i(t-1) + A_i(t) \big) \big(1 - \mathbb{I}\{\rho_i(t) > 0\}\big),
\end{align}
for all $i \in \mathcal{M}$ and \( t \in \{1,2,\ldots\} \), with \( \Delta_i(0) = 0, \forall i \). 
The long-term expected average VAoI is $\lim_{T\to\infty}\frac{1}{T}\sum_{t=1}^T\sum_{i=1}^M w_i\,\mathbb{E}[\Delta_i(t)]$, where $(w_1, w_2, \ldots, w_M)$ are user weights.

Our goal is to design a scheduling policy that jointly selects transmitting bits and power per user in each slot to minimize the long-term expected average VAoI subject to average power and distortion constraints. Specifically, we solve the following optimization problem.
\begin{subequations}\label{eq:main-opt-problem}
\begin{align}
V_{\rm opt} = \underset{\phi}{\text{min}} \;&
\lim_{T \to \infty} \frac{1}{T} 
\sum_{t=1}^{T} \sum_{i=1}^{M} 
w_i\, \mathbb{E}\left[\Delta_i(t)\right], \\
\text{subject to} \;\; &
\lim_{T \to \infty} \frac{1}{T} 
\sum_{t=1}^{T} \mathbb{E}\big[f_i(t)\big]
\le \bar{P}_i, \;\; \forall i \in \mathcal{M}, \label{eq:pow_constraint}\\
&
\lim_{T \to \infty} \frac{1}{T} 
\sum_{t=1}^{T} \mathbb{E}\big[d_i(t)\big]
\le \bar{D}_i, \;\; \forall i \in \mathcal{M}, \label{eq:dist_constraint}\\
& \eqref{eq:MAC}, \;\eqref{eq:age_evolution}. \notag
\end{align}
\end{subequations}
Here, $\phi$ denotes the scheduling policy. The parameters $\bar{P}_i$ and $\bar{D}_i$ denote the bounds on the average transmit power and the average distortion, respectively. All expectations are taken over the randomness in the system under policy $\phi$.

We also consider a TDMA benchmark (which allows at most one user to transmit in each slot), obtained by solving~\eqref{eq:main-opt-problem} with the following additional constraint:
\begin{align*}
    \sum_{i=1}^{M} \mathbb{I}\{\rho_i(t) > 0\} \leq 1, \quad \forall t \in \{1, 2, \ldots, T\}.
\end{align*}

\section{Optimization Analysis}
We study VAoI-agnostic stationary randomized policies (VA-SRPs) to solve \eqref{eq:main-opt-problem}, derive the optimal policy, and obtain a performance bound relative to $V_{\rm opt}$.

\subsection{VAoI-Agnostic Stationary Randomized Polices (VA-SRP)}
We define the class of VA-SRPs, reformulate \eqref{eq:main-opt-problem} under this class, and solve the resulting problem.
\subsubsection{Problem Reformulation}
Define
$\boldsymbol{\mathcal{R}} \triangleq \{ (\rho_1, \rho_2, \ldots, \rho_M) \allowbreak \mid \rho_i \in \{0, 1, \ldots, r_{\rm max}\}, \forall i \in \mathcal{M} \} 
= \{0, 1, \ldots, r_{\rm max}\}^M$.

{Policy (VA-SRP):} 
\emph{In each slot, when channel state is $\mathbf{h}$, the policy selects a rate vector $\boldsymbol{\rho} \triangleq (\rho_1, \dots, \rho_M) \in \boldsymbol{\mathcal{R}}$ with probability $\mu(\mathbf{h}, \boldsymbol{\rho})$, where $\rho_i \in \{0, 1, \ldots, r_{\max}\}$ specifies the number of bits transmitted by User~$i$. The corresponding transmit-powers vector $\boldsymbol{f}(\mathbf{h}, \boldsymbol{\rho}) \triangleq (f_i(\mathbf{h}, \boldsymbol{\rho}))_{i \in \mathcal{M}}$ satisfies the MAC constraints.}

This policy does not require knowledge of instantaneous VAoI, simplifying implementation. We next derive expressions for the long-term average VAoI, power, and distortion.

\paragraph{Long-term Expected Average VAoI under VA-SRP}
Under VA-SRP, the conditional probability of a successful status update delivery from User~$i$ given channel state $\mathbf{h}$ is
$\sum_{\boldsymbol{\rho}\in\boldsymbol{\mathcal{R}}}\mu(\mathbf{h},\boldsymbol{\rho})\mathbb{I}\{\rho_i>0\}$. By \eqref{eq:age_evolution}, $\rho_i>0$ implies a successful update and VAoI reset. Thus, the per-slot success probability for User~$i$ is $p_i(\boldsymbol{\mu})
= \sum_{\boldsymbol{\rho}\in\boldsymbol{\mathcal{R}}}\mathbb{E}_{\mathbf{h}}\!\left[\mu(\mathbf{h},\boldsymbol{\rho})\mathbb{I}\{\rho_i>0\}\right]$,
where $\boldsymbol{\mu}\triangleq(\mu(\mathbf{h},\boldsymbol{\rho}))_{\mathbf{h},\boldsymbol{\rho}}$.

Under VA-SRP, the VAoI process forms a Markov chain since the policy depends only on the channel state and arrivals and channels are i.i.d. Although \cite{Gangadhar_VAoI} considers a broadcast channel, the per-user VAoI evolution under an SRP has the same Markov structure (see Fig.~3 in \cite{Gangadhar_VAoI}). Consequently, the objective function in \eqref{eq:main-opt-problem} under VA-SRP can be obtained from the stationary distribution of this Markov chain. Applying~\cite[Theorem~1]{Gangadhar_VAoI}, the long-term average VAoI for User~$i$ is
\begin{align}\label{eq:SRP_VAoI}
    \lim_{T \to \infty} \frac{1}{T} \sum_{t=1}^{T} \mathbb{E}[\Delta_i(t)] 
    = \lambda_i \!\left(\frac{1}{p_i(\boldsymbol{\mu})} - 1 \right), 
    \quad \forall i \in \mathcal{M}.
\end{align}

\paragraph{Long-term Expected Average Transmit Power under VA-SRP}
Let $P_i(\boldsymbol{\mu}, \boldsymbol{F})$ denote the expected average transmit power of User~$i$, which under NOMA depends on all users’ transmission decisions. It is given by
$
P_i(\boldsymbol{\mu}, \boldsymbol{F})
= \sum_{\boldsymbol{\rho}\in\boldsymbol{\mathcal{R}}}
\mathbb{E}_{\mathbf{h}}\!\left[\mu(\mathbf{h},\boldsymbol{\rho}) f_i(\mathbf{h},\boldsymbol{\rho})\right],
$
where $f_i(\mathbf{h},\boldsymbol{\rho})$ is the power required to transmit $\rho_i$ bits under joint transmission $\boldsymbol{\rho}$ and channel state $\mathbf{h}$. Here, $\boldsymbol{F}\triangleq (f_i(\mathbf{h},\boldsymbol{\rho}))_{i,\mathbf{h},\boldsymbol{\rho}}$ denotes the collection of power allocations. The average power constraints under VA-SRP are $P_i(\boldsymbol{\mu},\boldsymbol{F})\le\bar{P}_i$, for all $i\in\mathcal{M}$.

\paragraph{Long-term Expected Average Distortion under VA-SRP}
By the law of total expectation, the per-slot expected distortion for User~$i$ is $\mathbb{E}[d_i(t)]= \mathbb{E}[d_i(t)\mid Q_i(t)=1]\mathbb{P}(Q_i(t)=1)$, since $d_i(t)=0$ when $Q_i(t)=0$. Conditioned on a pending packet, the expected per-slot distortion under VA-SRP is $\mathbb{E}[d_i(t)\mid Q_i(t)=1]
= \sum_{\boldsymbol{\rho}\in\boldsymbol{\mathcal{R}}}
\mathbb{E}_{\mathbf{h}}\!\left[\mu(\mathbf{h},\boldsymbol{\rho})\delta(\rho_i)\mathbb{I}\{\rho_i>0\}\right]$. Thus, $\mathbb{E}[d_i(t)]
= \mathbb{P}(Q_i(t)=1)
\sum_{\boldsymbol{\rho}\in\boldsymbol{\mathcal{R}}}
\mathbb{E}_{\mathbf{h}}\!\left[\mu(\mathbf{h},\boldsymbol{\rho})\delta(\rho_i)\mathbb{I}\{\rho_i>0\}\right]$.

The steady-state probability of a pending packet for User~$i$ follows from~\cite[Theorem~1]{Gangadhar_VAoI} as
\begin{align}\label{eq:Q_prob}
    \mathbb{P}(Q_i(t)=1) = \frac{\lambda_i}{\lambda_i \left(1-p_i(\boldsymbol{\mu})\right) + p_i(\boldsymbol{\mu})}, \quad \forall i \in \mathcal{M},
\end{align}
where $\lambda_i$ is the packet arrival probability and $p_i(\boldsymbol{\mu})$ is the status update delivery probability. This occupancy probability depends only on arrivals and transmissions, and is independent of power and distortion constraints. Substituting yields the long-term average distortion for User~$i$
\begin{align*}
D_i(\boldsymbol{\mu})=\frac{\lambda_i\sum_{\boldsymbol{\rho} \in \boldsymbol{\mathcal{R}}} \mathbb{E}_{\mathbf{h}}\big[ \mu(\mathbf{h},\boldsymbol{\rho}) \, \delta(\rho_i)\mathbb{I}\{\rho_i > 0\} \big]}{\lambda_i \left(1-p_i(\boldsymbol{\mu})\right) + p_i(\boldsymbol{\mu})}, \quad \forall i \in \mathcal{M}.
\end{align*}
Thus, the distortion constraints under VA-SRP, $D_i(\boldsymbol{\mu})\le\bar{D}_i$, are equivalently written as
\begin{align*}
    D'_i(\boldsymbol{\mu}) 
    &\triangleq \lambda_i \sum_{\boldsymbol{\rho} \in \boldsymbol{\mathcal{R}}} \mathbb{E}_{\mathbf{h}}\big[ \mu(\mathbf{h}, \boldsymbol{\rho}) \, \delta(\rho_i) \mathbb{I}\{\rho_i > 0\} \big]\nonumber\\
    &\le \bar{D}'_i(\boldsymbol{\mu}) \triangleq \bar{D}_i \, \Big(\lambda_i - \lambda_i p_i(\boldsymbol{\mu}) + p_i(\boldsymbol{\mu})\Big), \; \forall i \in \mathcal{M}.
\end{align*}

\paragraph{Reformulation of \eqref{eq:main-opt-problem} under VA-SRP}
Using \eqref{eq:SRP_VAoI} and the power and distortion constraints, \eqref{eq:main-opt-problem} is reformulated under VA-SRPs. To expose its convex structure, we introduce auxiliary variables
\begin{align}\label{eq:pi_def}
\pi_i(\mathbf{h}, \boldsymbol{\rho}) \triangleq \mu(\mathbf{h}, \boldsymbol{\rho}) \cdot f_i(\mathbf{h}, \boldsymbol{\rho}),
\end{align}
for all $i \in \mathcal{M}, \; \mathbf{h} \in \boldsymbol{\mathcal{H}}, \; \boldsymbol{\rho} \in \boldsymbol{\mathcal{R}}$,
representing the expected power contribution of User~$i$ under channel state $\mathbf{h}$ and rate vector $\boldsymbol{\rho}$. The reformulated problem is
\begin{subequations}\label{eq:SRP_pi}
\begin{align}
V_{\rm SRP} = \min_{\boldsymbol{\mu}, \boldsymbol{\pi}} \; & \sum_{i=1}^{M} w_i \lambda_i \left(\frac{1}{p_i(\boldsymbol{\mu})} - 1\right),\label{eq:obj_pi}\\
\text{subject to} \;\; & \sum_{\boldsymbol{\rho} \in \boldsymbol{\mathcal{R}}} \mathbb{E}_{\mathbf{h}}[\pi_i(\mathbf{h}, \boldsymbol{\rho})] \leq \bar{P}_i, \quad \forall i \in \mathcal{M}, \label{eq:pow_pi}\\
& D'_i(\boldsymbol{\mu}) \leq \bar{D}'_i(\boldsymbol{\mu}), \quad \forall i \in \mathcal{M}, \label{eq:dist_pi}\\
& \mu(\mathbf{h}, \boldsymbol{\rho}) \cdot g^{-1}\left(\sum_{i \in \mathcal{S}} \rho_i\right) \leq \sum_{i \in \mathcal{S}} \pi_i(\mathbf{h}, \boldsymbol{\rho}) h_i, \notag\\
& \qquad \qquad \forall \mathcal{S} \subseteq \mathcal{M}, \; \mathbf{h} \in \boldsymbol{\mathcal{H}}, \; \boldsymbol{\rho} \in \boldsymbol{\mathcal{R}}, \label{eq:MAC_pi}\\
& 0 \leq \mu(\mathbf{h}, \boldsymbol{\rho}) \leq 1, \quad \forall \mathbf{h} \in \boldsymbol{\mathcal{H}}, \; \boldsymbol{\rho} \in \boldsymbol{\mathcal{R}}, \label{eq:mu_bound}\\
& \sum_{\boldsymbol{\rho} \in \boldsymbol{\mathcal{R}}} \mu(\mathbf{h}, \boldsymbol{\rho}) = 1, \quad \forall \mathbf{h} \in \boldsymbol{\mathcal{H}}, \label{eq:mu_sum}
\end{align}
\end{subequations}
where $\boldsymbol{\pi}\triangleq(\pi_i(\mathbf{h},\boldsymbol{\rho}))_{i,\mathbf{h},\boldsymbol{\rho}}$.
The MAC constraint~\eqref{eq:MAC_pi} follows from~\eqref{eq:MAC} since $g(\cdot)$ is increasing:
$g^{-1}\!\left(\sum_{i\in\mathcal{S}}\rho_i\right) \le \sum_{i\in\mathcal{S}} f_i(\mathbf{h},\boldsymbol{\rho})h_i$,
which, after multiplying by $\mu(\mathbf{h},\boldsymbol{\rho})\ge0$ and using~\eqref{eq:pi_def}, yields~\eqref{eq:MAC_pi}.

Given the optimal $(\boldsymbol{\mu}^*,\boldsymbol{\pi}^*)$ for~\eqref{eq:SRP_pi}, we have 
\begin{align}\label{eq:F_from_pi}
f_i^*(\mathbf{h}, \boldsymbol{\rho}) =
\begin{cases}
\dfrac{\pi_i(\mathbf{h}, \boldsymbol{\rho})}{\mu(\mathbf{h}, \boldsymbol{\rho})}, & \mu(\mathbf{h}, \boldsymbol{\rho}) > 0, \\
0, & \mu(\mathbf{h}, \boldsymbol{\rho}) = 0,
\end{cases}
\end{align}
for all $i \in \mathcal{M}$ and $(\mathbf{h}, \boldsymbol{\rho})$.

\begin{theorem}\label{thm:convexity} \eqref{eq:SRP_pi} is a convex optimization problem.
\end{theorem}
\begin{proof}
    See Appendix A.
\end{proof}

Despite convexity, solving~\eqref{eq:SRP_pi} directly using standard solvers is computationally challenging due to the exponential number of MAC constraints. For each $(\mathbf{h},\boldsymbol{\rho})$ pair, there are $2^M-1$ constraints corresponding to all non-empty subsets $\mathcal{S}\subseteq\mathcal{M}$, yielding a total of $|\boldsymbol{\mathcal{H}}|\,|\boldsymbol{\mathcal{R}}|\,(2^M-1)$ constraints. This exponential scaling renders direct methods infeasible even for moderate $M$. To overcome this, we employ a Lagrangian dual decomposition that exploits the problem structure and significantly reduces complexity.

\subsubsection{Lagrangian Dual Formulation}
Due to the convexity of~\eqref{eq:SRP_pi}, we solve the problem via Lagrangian duality.
Slater’s condition holds provided the problem is feasible. In particular, there exists
$\varepsilon>0$ such that a randomized policy with strictly positive probabilities
$\mu(\mathbf{h},\boldsymbol{\rho})>0$ for all $(\mathbf{h},\boldsymbol{\rho})$ and sufficiently small
$\pi_i(\mathbf{h},\boldsymbol{\rho})=\varepsilon$ strictly satisfies the power, distortion, and MAC
constraints, while maintaining $0<\mu(\mathbf{h},\boldsymbol{\rho})<1$ and
$\sum_{\boldsymbol{\rho}}\mu(\mathbf{h},\boldsymbol{\rho})=1$.
Hence, Slater’s condition is satisfied and strong duality holds~\cite{Boyd}.
To enable decomposition, we reformulate the problem so that the Lagrangian is linear in
$\boldsymbol{\mu}$.

We introduce auxiliary variables $\boldsymbol{\eta} \triangleq (\eta_i)_{i\in\mathcal{M}}$ to linearize the objective by replacing $1/p_i(\boldsymbol{\mu})$, with linking constraints $p_i(\boldsymbol{\mu})\ge 1/\eta_i$. 
The resulting
equivalent formulation is
\begin{subequations}\label{eq:SRP_eta}
\begin{align}
V_{\rm SRP} = \min_{\boldsymbol{\mu}, \boldsymbol{\pi}, \boldsymbol{\eta}} \; & \sum_{i=1}^{M} w_i \lambda_i (\eta_i - 1), \label{eq:obj_eta}\\
\text{subject to} \;\; & p_i(\boldsymbol{\mu}) \geq \frac{1}{\eta_i}, \quad \forall i \in \mathcal{M}, \label{eq:eta_link}\\
&D'_i(\boldsymbol{\mu}) \leq \bar{D}'_i(\eta_i), \quad \forall i \in \mathcal{M},\label{eq:dist_const_eta}\\
&\eqref{eq:pow_pi}, \eqref{eq:MAC_pi}, \eqref{eq:mu_bound}, \eqref{eq:mu_sum},\nonumber
\end{align}
\end{subequations}
where $\bar{D}'_i(\eta_i) = \bar{D}_i (\lambda_i - \lambda_i/ \eta_i + 1/\eta_i), \; \forall i \in \mathcal{M}$.
Since $w_i\lambda_i>0$, constraint~\eqref{eq:eta_link} binds at optimality, yielding
$\eta_i = 1/p_i(\boldsymbol{\mu})$ and ensuring equivalence between
\eqref{eq:SRP_eta} and~\eqref{eq:SRP_pi}.

We dualize the power, distortion, and linking constraints using Lagrange multipliers
$\boldsymbol{\beta} \triangleq (\beta)_{i \in \mathcal{M}}\ge \mathbf{0}$, $\boldsymbol{\alpha} \triangleq (\alpha)_{i \in \mathcal{M}} \ge \mathbf{0}$, and $\boldsymbol{\nu} \triangleq (\nu)_{i \in \mathcal{M}}  \ge \mathbf{0}$.
By convexity and Slater’s condition, strong duality holds, and the primal solutions
$(\boldsymbol{\mu}^*,\boldsymbol{\pi}^*,\boldsymbol{\eta}^*)$ are optimal if and only if they minimizes the Lagrangian for optimal multipliers
$(\boldsymbol{\beta}^*,\boldsymbol{\alpha}^*,\boldsymbol{\nu}^*)$. The corresponding Lagrangian is given by
\begin{align}\label{eq:Lagrangian_unified}
&L(\boldsymbol{\beta}, \boldsymbol{\alpha}, \boldsymbol{\nu}, \boldsymbol{\mu}, \boldsymbol{\pi}, \boldsymbol{\eta}) 
= \sum_{i=1}^{M} \Bigg(w_i \lambda_i (\eta_i - 1) - \nu_i \left(p_i(\boldsymbol{\mu}) - \frac{1}{\eta_i}\right)\nonumber\\
&\quad+  \beta_i \Big(\sum_{\boldsymbol{\rho} \in \boldsymbol{\mathcal{R}}} \mathbb{E}_{\mathbf{h}}[\pi_i(\mathbf{h}, \boldsymbol{\rho})] - \bar{P}_i\Big) + \alpha_i \left(D'_i(\boldsymbol{\mu}) - \bar{D}'_i(\eta_i)\right) \Bigg). \nonumber\\
\end{align}
The dual function is defined as
\begin{align*}
z(\boldsymbol{\beta}, \boldsymbol{\alpha}, \boldsymbol{\nu}) = \min_{\boldsymbol{\mu}, \boldsymbol{\pi}, \boldsymbol{\eta}} \; &L(\boldsymbol{\beta}, \boldsymbol{\alpha}, \boldsymbol{\nu}, \boldsymbol{\mu}, \boldsymbol{\pi}, \boldsymbol{\eta}),\\
\text{subject to} \;\; &\eqref{eq:MAC_pi}, \; \eqref{eq:mu_bound}, \; \eqref{eq:mu_sum}.\nonumber 
\end{align*}
The dual problem is
\begin{align}\label{eq:dual_problem}
     (\boldsymbol{\beta}^*, \boldsymbol{\alpha}^*, \boldsymbol{\nu}^*) = \arg\max_{\boldsymbol{\beta}\geq 0, \boldsymbol{\alpha} \geq 0, \boldsymbol{\nu} \geq 0} z(\boldsymbol{\beta}, \boldsymbol{\alpha}, \boldsymbol{\nu}).
\end{align}

Once dualized, the power, distortion, and linking constraints are enforced implicitly through the optimal multipliers. Solving the dual problem~\eqref{eq:dual_problem} yields the optimal multipliers
$(\boldsymbol{\beta}^*,\boldsymbol{\alpha}^*,\boldsymbol{\nu}^*)$ and the optimal dual objective
$z(\boldsymbol{\beta}^*,\boldsymbol{\alpha}^*,\boldsymbol{\nu}^*)$. By strong duality,
$V_{\rm SRP}=z(\boldsymbol{\beta}^*,\boldsymbol{\alpha}^*,\boldsymbol{\nu}^*)$, and any
$(\boldsymbol{\mu}^*,\boldsymbol{\pi}^*,\boldsymbol{\eta}^*)$ minimizing the Lagrangian
for optimal multipliers are the primal optimal solutions. For fixed $(\boldsymbol{\beta},\boldsymbol{\alpha},\boldsymbol{\nu})$, the dual function
is minimized by first optimizing $\boldsymbol{\pi}$ for fixed
$(\boldsymbol{\mu},\boldsymbol{\eta})$, followed by optimization over
$(\boldsymbol{\mu},\boldsymbol{\eta})$.

\paragraph{Power Allocation Optimization for Fixed $\boldsymbol{\mu}$ and $\boldsymbol{\eta}$} 
For fixed $(\boldsymbol{\mu},\boldsymbol{\eta})$, we optimize over the power variables $\boldsymbol{F}$, with $\boldsymbol{\pi}$ recovered via~\eqref{eq:pi_def}. The optimal $\boldsymbol{F}^*$ solves
\begin{align*}
    \underset{\boldsymbol{F}}{\text{minimize}} \; &\sum_{i = 1}^M \sum_{\mathbf{h} \in \boldsymbol{\mathcal{H}}} \mathbb{P}(\mathbf{h}) \sum_{\boldsymbol{\rho} \in \boldsymbol{\mathcal{R}}} \beta_i \mu(\mathbf{h}, \boldsymbol{\rho})f_i(\mathbf{h}, \boldsymbol{\rho}),\\
    \text{subject to} \;\; &\eqref{eq:MAC_pi},\nonumber 
\end{align*}
where $\mathbb{P}(\mathbf{h})$ denotes the probability that the channel state is $\mathbf{h}$.
Since the objective is separable across channel states, the problem decomposes into independent subproblems for each $(\mathbf{h},\boldsymbol{\rho})$, given by
\begin{align}\label{eq:power_min}
\underset{\boldsymbol{f}(\mathbf{h}, \boldsymbol{\rho})}{\text{minimize}} \; &\sum_{i=1}^{M} \beta_i f_i(\mathbf{h}, \boldsymbol{\rho}), \\
\text{subject to} \;\; &g^{-1}\left(\sum_{i \in \mathcal{S}} \rho_i\right) \leq \sum_{i \in \mathcal{S}} f_i(\mathbf{h}, \boldsymbol{\rho}) h_i, \quad \forall \mathcal{S} \subseteq \mathcal{M}.\nonumber
\end{align}

Problem~\eqref{eq:power_min} is independent of $\boldsymbol{\mu}$ and $\boldsymbol{\eta}$; hence, $\boldsymbol{f}^*(\mathbf{h},\boldsymbol{\rho})$ is identical for all feasible $(\boldsymbol{\mu},\boldsymbol{\eta})$, and states with $\mu(\mathbf{h},\boldsymbol{\rho})=0$ are immaterial, since such states do not contribute to either the
scheduling probability or the average power constraints. Problem~\eqref{eq:power_min} minimizes a weighted sum power over the MAC capacity region to achieve rate vector $\boldsymbol{\rho}$. Since the MAC region is a polymatroid and we minimizing
a linear objective (weighted sum power) over this region, the optimum is attained at a vertex, corresponding to a specific SIC decoding order.

\textit{NOMA-SIC Power Allocation:}
Under SIC, the BS decodes users sequentially according to a decoding order
$\boldsymbol{\theta} \triangleq (\theta_1,\ldots,\theta_M)$, where $\theta_j$ is index of the user decoded at
stage $j$. The required transmit powers for users, which satisfies the MAC constraints, given by
\begin{align}\label{eq:POW_SIC}
f_{\theta_j}(\mathbf{h}, \boldsymbol{\rho}, \boldsymbol{\theta})
= \frac{g^{-1}(\rho_{\theta_j})}{h_{\theta_j}}
\prod_{k=j+1}^{M} \left(1 + g^{-1}(\rho_{\theta_k})\right),
\quad \forall j \in \mathcal{M}.
\end{align}
with the convention that $g^{-1}(0)=0$ and the empty product equals~1.

To solve~\eqref{eq:power_min} and eliminate the MAC constraints, we consider the transmit power under NOMA--SIC and determine the decoding order that minimizes the weighted sum power
\begin{align}\label{eq:decoding_opt}
\boldsymbol{\theta}^*(\mathbf{h}, \boldsymbol{\rho}, \boldsymbol{\beta}) = \arg\min_{\boldsymbol{\theta}} \sum_{j=1}^{M} \beta_{\theta_j} f_{\theta_j}(\mathbf{h}, \boldsymbol{\rho}, \boldsymbol{\theta}).
\end{align}
A brute-force search over all $M!$ decoding orders is prohibitive; however, the optimal order can be found efficiently.

\begin{lemma}\label{lm:Lemma_decoding}
Given Lagrange multipliers $\boldsymbol{\beta}$, transmission rate vector $\boldsymbol{\rho}$, and channel state $\mathbf{h}$, the optimal decoding order $\boldsymbol{\theta}^*(\mathbf{h}, \boldsymbol{\rho}, \boldsymbol{\beta}) = (\theta_1^*, \ldots, \theta_M^*)$ that solves~\eqref{eq:decoding_opt} satisfies
\begin{align*}
\frac{h_{\theta_1^*} \mathbb{I}\{\rho_{\theta_1^*} > 0\}}{\beta_{\theta_1^*}} \geq \frac{h_{\theta_2^*} \mathbb{I}\{\rho_{\theta_2^*} > 0\}}{\beta_{\theta_2^*}} \geq \cdots \geq \frac{h_{\theta_M^*} \mathbb{I}\{\rho_{\theta_M^*} > 0\}}{\beta_{\theta_M^*}}.
\end{align*}
\end{lemma}

\begin{proof}
See Appendix~B.
\end{proof}

\begin{remark}
By Lemma~\ref{lm:Lemma_decoding}, the optimal decoding order is obtained by sorting the ratios $h_i \mathbb{I}\{\rho_i > 0\} / \beta_i$ in non-increasing order. This sorting operation has computational complexity $\mathcal{O}(M \log M)$ for each $(\mathbf{h}, \boldsymbol{\rho})$ pair, which is significantly lower than the $\mathcal{O}(M!)$ complexity of brute-force search.  
\end{remark}

Let $\boldsymbol{f}^*(\mathbf{h},\boldsymbol{\rho},\boldsymbol{\beta})$ denote the optimal NOMA-SIC power allocation, where the optimal decoding order depends on $\boldsymbol{\beta}$. Since~\eqref{eq:power_min} is independent of $(\boldsymbol{\mu},\boldsymbol{\eta})$, the optimal decoding order and power allocation are identical for all feasible $(\boldsymbol{\mu},\boldsymbol{\eta})$.
 Given any $\boldsymbol{\mu}$ and $\boldsymbol{\beta}$, we obtain the corresponding  $\boldsymbol{\pi}^*(\boldsymbol{\mu}, \boldsymbol{\beta})$ via
\begin{align*}
\pi_i^*(\mathbf{h}, \boldsymbol{\rho}) = \mu(\mathbf{h}, \boldsymbol{\rho}) \cdot f_i^*(\mathbf{h}, \boldsymbol{\rho}, \boldsymbol{\beta}),
\end{align*}
for all $i \in \mathcal{M}$ and $(\mathbf{h}, \boldsymbol{\rho})$.

\paragraph{Scheduling Probability and Auxiliary Variable Optimization}

Given the optimal power allocations $\boldsymbol{\pi}^*(\boldsymbol{\mu}, \boldsymbol{\beta})$, we optimize over $\boldsymbol{\mu}$ and $\boldsymbol{\eta}$. Substituting $\pi^*_i(\mathbf{h}, \boldsymbol{\rho})=\mu(\mathbf{h}, \boldsymbol{\rho}) f_i^*(\mathbf{h}, \boldsymbol{\rho}, \boldsymbol{\beta})$ into~\eqref{eq:Lagrangian_unified}, and using the
definitions of $p_i(\boldsymbol{\mu})$, $D'_i(\boldsymbol{\mu})$, and $\bar{D}'_i(\eta_i)$ yields
\begin{align}\label{eq:Lagrangian_expanded}
&L(\boldsymbol{\beta}, \boldsymbol{\alpha}, \boldsymbol{\nu}, \boldsymbol{\mu}, \boldsymbol{\pi}^*(\boldsymbol{\mu}), \boldsymbol{\eta}) \nonumber \\
&= \sum_{i=1}^{M} \left(w_i \lambda_i \eta_i + \frac{\nu_i}{\eta_i} - \frac{\alpha_i\bar{D}_i  (1-\lambda_i)}{\eta_i}\right) \nonumber\\
&\quad + \sum_{\mathbf{h} \in \boldsymbol{\mathcal{H}}} \sum_{\boldsymbol{\rho} \in \boldsymbol{\mathcal{R}}} \mathbb{P}(\mathbf{h}) \mu(\mathbf{h}, \boldsymbol{\rho}) \Bigg[\sum_{i=1}^{M} \Big(\beta_i f_i^*(\mathbf{h}, \boldsymbol{\rho}, \boldsymbol{\beta}) \nonumber\\
&\quad + \alpha_i \lambda_i \delta(\rho_i) \mathbb{I}\{\rho_i > 0\} - \nu_i \mathbb{I}\{\rho_i > 0\}\Big)\Bigg] + C,
\end{align}
where $C = -\sum_{i=1}^{M} (w_i \lambda_i + \beta_i \bar{P}_i + \alpha_i \bar{D}_i \lambda_i)$ is a constant independent of $\boldsymbol{\mu}$ and $\boldsymbol{\eta}$.  The dual function is
\begin{align*}
z(\boldsymbol{\beta}, \boldsymbol{\alpha}, \boldsymbol{\nu}) = \min_{\boldsymbol{\mu},\boldsymbol{\eta}} \; &L(\boldsymbol{\beta}, \boldsymbol{\alpha}, \boldsymbol{\nu}, \boldsymbol{\mu}, \boldsymbol{\pi}^*(\boldsymbol{\mu}, \boldsymbol{\beta}), \boldsymbol{\eta}),\\
\text{subject to} \;\; &\eqref{eq:mu_bound}, \; \eqref{eq:mu_sum}.\nonumber 
\end{align*}
As~\eqref{eq:Lagrangian_expanded} is separable in $\boldsymbol{\mu}$ and $\boldsymbol{\eta}$, the minimization decouples.

\textit{Optimization over $\boldsymbol{\mu}$:} 
The terms involving $\mu(\mathbf{h}, \boldsymbol{\rho})$ in~\eqref{eq:Lagrangian_expanded} further decompose across channel states $\mathbf{h}$,  since the constraints $0 \leq \mu(\mathbf{h}, \boldsymbol{\rho}) \leq 1,\; \forall \mathbf{h}, \boldsymbol{\rho}$ and $\sum_{\boldsymbol{\rho}} \mu(\mathbf{h}, \boldsymbol{\rho}) = 1,\; \forall \mathbf{h}$ are imposed independently for each $\mathbf{h}$. Hence, for each $(\mathbf{h}, \boldsymbol{\rho})$, $\mu^*(\mathbf{h}, \boldsymbol{\rho})$ is optimal if and only if it solution to the following subproblem
\begin{align}\label{eq:mu_h_subproblem}
\underset{\mu(\mathbf{h}, \boldsymbol{\rho}), \forall \boldsymbol{\rho}}{\text{minimize}} \; &\sum_{\boldsymbol{\rho} \in \boldsymbol{\mathcal{R}}} \mu(\mathbf{h}, \boldsymbol{\rho}) A(\mathbf{h}, \boldsymbol{\rho}),\\
\text{subject to} \;\; &\sum_{\boldsymbol{\rho} \in \boldsymbol{\mathcal{R}}} \mu(\mathbf{h}, \boldsymbol{\rho}) = 1, \quad 0 \leq \mu(\mathbf{h}, \boldsymbol{\rho}) \leq 1, \; \forall \boldsymbol{\rho}.\nonumber
\end{align}
Here $A(\mathbf{h}, \boldsymbol{\rho}) = \sum_{i=1}^{M} \Big(\beta_i f_i^*(\mathbf{h}, \boldsymbol{\rho}, \boldsymbol{\beta}) + \alpha_i \lambda_i \delta(\rho_i) \mathbb{I}\{\rho_i > 0\} - \nu_i \mathbb{I}\{\rho_i > 0\}\Big)$. Problem~\eqref{eq:mu_h_subproblem} is a linear program with a closed-form solution.
Let $J^*(\mathbf{h}) = \{\boldsymbol{\rho} \in \boldsymbol{\mathcal{R}} : A(\mathbf{h}, \boldsymbol{\rho}) = \min_{\boldsymbol{\rho}'} A(\mathbf{h}, \boldsymbol{\rho}')\}$ denote the set of minimizers. Then, for each $(\mathbf{h}, \boldsymbol{\rho})$ pair, the optimal $\mu^*(\mathbf{h}, \boldsymbol{\rho})$ is 
\begin{align}\label{eq:mu_optimal}
\mu^*(\mathbf{h}, \boldsymbol{\rho}) = 
\begin{cases}
1, & \text{if } |J^*(\mathbf{h})| = 1 \text{ and } \boldsymbol{\rho} \in J^*(\mathbf{h}),\\
1/|J^*(\mathbf{h})|, & \text{if } |J^*(\mathbf{h})| > 1 \text{ and } \boldsymbol{\rho} \in J^*(\mathbf{h}),\\
0, & \text{otherwise}.
\end{cases}
\end{align}
When multiple minimizers exist, $\mu^*(\mathbf{h}, \boldsymbol{\rho})$ is chosen as the uniform distribution over $J^*(\mathbf{h})$.

\textit{Optimization over $\boldsymbol{\eta}$:}
The terms in~\eqref{eq:Lagrangian_expanded} involving $\eta_i$ are separable across $i$. Minimizing with respect to $\eta_i$ yields
\begin{align*}
\frac{\partial L}{\partial \eta_i} = w_i \lambda_i   -\frac{\nu_i - \alpha_i \bar{D}_i (1-\lambda_i)}{\eta_i^2} = 0.
\end{align*}
Solving for $\eta_i$ gives
\begin{align*}
\eta_i^* = \sqrt{\frac{\nu_i - \alpha_i \bar{D}_i (1-\lambda_i)}{w_i \lambda_i}}.
\end{align*}

For given $(\boldsymbol{\beta},\boldsymbol{\alpha},\boldsymbol{\nu})$, the optimal primal variables
$(\boldsymbol{\mu}^*,\boldsymbol{\pi}^*,\boldsymbol{\eta}^*)$ are obtained; arguments are omitted for brevity.

The dual problem is therefore
\begin{align}\label{eq:dual_problem_2}
(\boldsymbol{\beta}^*, \boldsymbol{\alpha}^*, \boldsymbol{\nu}^*)
= \arg\max_{\boldsymbol{\beta} \geq 0,\;
\boldsymbol{\alpha} \geq 0,\;
\boldsymbol{\nu} \geq 0}
z(\boldsymbol{\beta}, \boldsymbol{\alpha}, \boldsymbol{\nu}),
\end{align}
where $z(\boldsymbol{\beta},\boldsymbol{\alpha},\boldsymbol{\nu})
= L(\boldsymbol{\beta},\boldsymbol{\alpha},\boldsymbol{\nu},
\boldsymbol{\mu}^*,\boldsymbol{\pi}^*,\boldsymbol{\eta}^*)$.
By strong duality, $V_{\rm SRP}=z(\boldsymbol{\beta}^*,\boldsymbol{\alpha}^*,\boldsymbol{\nu}^*)$, and the corresponding Lagrangian minimizers solve~\eqref{eq:SRP_pi}. Next, we describe a subgradient ascent algorithm to solve~\eqref{eq:dual_problem_2}.

\paragraph{Subgradient Ascent Algorithm for Solving the Dual Problem}
The dual function $z(\boldsymbol{\beta},\boldsymbol{\alpha},\boldsymbol{\nu})$ is concave but generally nonsmooth~\cite{Boyd}; it is not clear that it is differentiable everywhere. Hence, we employ a subgradient ascent method~\cite{Shor1985,Bertsekas2016}.

The subgradients of the dual function with respect to the dual variables
are given by
\begin{align*}
\zeta_{\beta_i} &= \left(\sum_{\boldsymbol{\rho} \in \boldsymbol{\mathcal{R}}}
\mathbb{E}_{\mathbf{h}}\!\left[\pi_i(\mathbf{h}, \boldsymbol{\rho})\right]
- \bar{P}_i\right) \in \partial_{\beta_i} z(\boldsymbol{\beta}, \boldsymbol{\alpha}, \boldsymbol{\nu}), \\
\zeta_{\alpha_i} &= \left(D'_i(\boldsymbol{\mu}) - \bar{D}'_i(\eta_i)\right) \in \partial_{\alpha_i} z(\boldsymbol{\beta}, \boldsymbol{\alpha}, \boldsymbol{\nu}), \\
\zeta_{\nu_i} &= \left(\frac{1}{\eta_i} - p_i(\boldsymbol{\mu})\right) \in \partial_{\nu_i} z(\boldsymbol{\beta},  \boldsymbol{\alpha}, \boldsymbol{\nu}),
\end{align*}
for all $i \in \mathcal{M}$, where $\partial z(\cdot)$ denotes the set of subgradients (i.e., the subdifferential) of the dual function $z(\cdot)$ at the dual variables.
These subgradients correspond to violations of the power, distortion, and linking constraints, respectively. Starting from any nonnegative initial dual variables
$(\boldsymbol{\beta}_0, \boldsymbol{\alpha}_0, \boldsymbol{\nu}_0)$,
the dual variables are updated via subgradient ascent iterations
\begin{subequations}\label{eq:subgradient_update}
\begin{align}
\beta_{i, k+1} &= \left[\beta_{i,k} + s_k \left(\sum_{\boldsymbol{\rho} \in \boldsymbol{\mathcal{R}}} \mathbb{E}_{\mathbf{h}}[\pi_{i,k}^{*}(\mathbf{h}, \boldsymbol{\rho})] - \bar{P}_i\right)\right]^+,\label{eq:beta_update}\\
\alpha_{i,k+1} &= \left[\alpha_{i,k} + s_k \left(D'_i(\boldsymbol{\mu}^*_k) - \bar{D}'_i(\eta^*_{i,k})\right)\right]^+, \label{eq:alpha_update}\\
\nu_{i,k+1} &= \left[\nu_{i,k} + s_k \left(\frac{1}{\eta^*_{i,k}} - p_i(\boldsymbol{\mu}^*_k)\right)\right]^+, \label{eq:nu_update}
\end{align}
\end{subequations}
for all $i \in \mathcal{M}$ and $k$. 
Here, $k$ is the iteration index, $[\cdot]^+=\max\{0,\cdot\}$, and $s_k>0$ is a diminishing step size satisfying $\sum_k s_k=\infty$ and $\sum_k s_k^2<\infty$ (e.g., $s_k=1/k$).

At each iteration $k$, the primal variables $(\boldsymbol{\mu}^*_k,\boldsymbol{\pi}^*_k,\boldsymbol{\eta}^*_k)$ minimize the Lagrangian for $(\boldsymbol{\beta}_k,\boldsymbol{\alpha}_k,\boldsymbol{\nu}_k)$. Under the above step-size conditions, the iterates converge to $(\boldsymbol{\beta}^*,\boldsymbol{\alpha}^*,\boldsymbol{\nu}^*)$, and by strong duality the corresponding Lagrangian minimizers solve~\eqref{eq:SRP_pi}~\cite{Boyd,Shor1985,Bertsekas2016}.

To evaluate the terms $1/\eta^*_{i,k}$ appearing in \eqref{eq:nu_update} and
\eqref{eq:alpha_update} (through $\bar{D}'_i(\eta^*_{i,k})$), it is necessary to
ensure that $\eta^*_{i,k}>0$ for all $i$. To this end, we employ an
$\epsilon_r$-regularized primal update to avoid degeneracy of $\eta^*_{i,k}$,
where $\epsilon_r>0$ is a small regularization parameter. Specifically,
$\eta^*_{i,k}$ is computed as
\begin{align}\label{eq:eta_opt_reg}
\eta^*_{i,k}
=
\sqrt{
\frac{
\max\!\left\{
\nu_{i,k} - \alpha_{i,k}\,\bar{D}_i(1-\lambda_i),\;
\epsilon_r
\right\}
}{
w_i \lambda_i
}
}, \quad \forall i, k.
\end{align}

For any fixed $\epsilon_r>0$, the regularized dual problem remains convex and converges under standard diminishing step sizes, with an optimality gap that vanishes as $\epsilon_r\!\to\!0$. Alternatively, a vanishing regularization $\epsilon_{r,k} \propto s_k$ preserves convergence to the unregularized optimum.

Algorithm~\ref{alg:SRP} summarizes the procedure for solving~\eqref{eq:SRP_pi}.

\begin{algorithm}[t]
\caption{Subgradient-based solution of~\eqref{eq:SRP_pi}}
\label{alg:SRP}
\begin{algorithmic}[1]
\Require $M$, $r_{\rm max}$, $(\lambda_i, w_i, \bar{P}_i, \bar{D}_i)_{i \in \mathcal{M}}$, $\boldsymbol{\mathcal{H}}$, $\mathbb{P}(\mathbf{h})$ for all $\mathbf{h} \in \boldsymbol{\mathcal{H}}$, tolerance $\epsilon > 0$, $\epsilon_r > 0 $
    
\State Initialize iteration index $k = 0$, dual variables $\boldsymbol{\beta}_k \geq \mathbf{0}$, $\boldsymbol{\alpha}_k \geq \mathbf{0}, \boldsymbol{\nu}_k  \geq \mathbf{0}$

\Repeat 
    \State \parbox[t]{\dimexpr\linewidth-\algorithmicindent}{For each $(\mathbf{h},\boldsymbol{\rho})$,
         compute optimal decoding order $\boldsymbol{\theta}^*(\mathbf{h},\boldsymbol{\rho},\boldsymbol{\beta}_{k})$ 
        by sorting $(h_i \mathbb{I}\{\rho_i > 0\}/\beta_{i,k})_{i \in \mathcal{M}}$ 
        in non-increasing order,
    and obtain the optimal power allocation $\boldsymbol{f}^*(\mathbf{h},\boldsymbol{\rho},\boldsymbol{\beta}_{k})$ 
        using~\eqref{eq:POW_SIC}   \label{step:loop}.}

    \State For each $(\mathbf{h}, \boldsymbol{\rho})$, compute $\mu_{k}^*(\mathbf{h}, \boldsymbol{\rho})$ using \eqref{eq:mu_optimal}.
    \State \parbox[t]{\dimexpr\linewidth-\algorithmicindent}{For each $i \in \mathcal{M}$, compute $p_i(\boldsymbol{\mu}_k^*)$, $D'_i(\boldsymbol{\mu}_k^{*})$, $\eta^*_{i, k}$ using \eqref{eq:eta_opt_reg}, and $\bar{D}'_i(\eta_{i,k}^*)$.}
    \State \parbox[t]{\dimexpr\linewidth-\algorithmicindent}{For each $i \in \mathcal{M}$, update $\beta_{i, k+1}, \alpha_{i, k+1}$ and $\nu_{i,k+1}$ using~\eqref{eq:subgradient_update}.}
    \State $k \leftarrow k + 1$
        
    \Until{$|z(\boldsymbol{\beta}_k, \boldsymbol{\alpha}_k, \boldsymbol{\nu}_k) - z(\boldsymbol{\beta}_{k+1}, \boldsymbol{\alpha}_{k+1}, \boldsymbol{\nu}_{k+1})| < \epsilon$}

    \State At convergence, set $\boldsymbol{\beta}^* \leftarrow \boldsymbol{\beta}_k$, $\boldsymbol{\alpha}^* \leftarrow \boldsymbol{\alpha}_{k}, \boldsymbol{\nu}^* \leftarrow \boldsymbol{\nu}_{k}$, and obtain $\boldsymbol{\mu}^*(\boldsymbol{\beta}^*, \boldsymbol{\alpha}^*, \boldsymbol{\nu}^*)$ using \eqref{eq:mu_optimal}.
    \State For each $i, (\mathbf{h}, \boldsymbol{\rho})$, compute $\pi_i^*(\mathbf{h}, \boldsymbol{\rho}) = \mu^*(\mathbf{h}, \boldsymbol{\rho}) \cdot f_i^*(\mathbf{h}, \boldsymbol{\rho}, \boldsymbol{\beta}^*)$, and thereby obtain $\boldsymbol{\pi}^*(\boldsymbol{\beta}^*, \boldsymbol{\alpha}^*, \boldsymbol{\nu}^*)$.

\Ensure Optimal solutions for \eqref{eq:SRP_pi}: $ \boldsymbol{\mu}^*(\boldsymbol{\beta}^*, \boldsymbol{\alpha}^*, \boldsymbol{\nu}^*)$ and $\boldsymbol{\pi}^*(\boldsymbol{\beta}^*, \boldsymbol{\alpha}^*, \boldsymbol{\nu}^*)$. Obtain $\boldsymbol{F}^*$ via \eqref{eq:F_from_pi}.
\end{algorithmic}
\end{algorithm}

\begin{remark}[Dual Variable Convergence and Averaging]\label{rem:averaging}
When the dual problem~\eqref{eq:dual_problem_2} admits multiple optimal
solutions, the subgradient method may not converge to a unique dual
point; instead, the dual iterates may oscillate among optimal solutions
while attaining the optimal dual value. In such cases, the primal
solutions corresponding to individual dual iterates may violate some
primal constraints, no single dual
solution necessarily enforces all primal constraints simultaneously;
rather, appropriate convex combinations of the associated 
solutions ($\boldsymbol{\mu}^*, \boldsymbol{\pi^*}$) are required.
This behavior is well known in dual decomposition and subgradient-based
methods~\cite{Nedic2009}. 

Nevertheless, under standard diminishing step-size conditions, the \emph{ergodic (time-averaged) primal iterates}
\begin{align*}
\bar{\boldsymbol{\mu}}_L &= \frac{1}{L}\sum_{l=1}^{L} \boldsymbol{\mu}^*(\boldsymbol{\beta}^*_l, \boldsymbol{\alpha}^*_l, \boldsymbol{\nu}^*_l),\;\;
\bar{\boldsymbol{\pi}}_L = \frac{1}{L}\sum_{l=1}^{L} \boldsymbol{\pi}^*(\boldsymbol{\beta}^*_l, \boldsymbol{\alpha}^*_l, \boldsymbol{\nu}^*_l),
\end{align*}  
are guaranteed to converge to the optimal primal solutions and satisfy all
primal constraints, including MAC constraints~\cite{Nedic2009}.
Here $L$ is the averaging window (e.g., $L = 500$).
Specifically, once the dual iterates are deemed to have converged in the
sense that $|z(\boldsymbol{\beta}_k, \boldsymbol{\alpha}_k, \boldsymbol{\nu}_k) - z(\boldsymbol{\beta}_{k+1}, \boldsymbol{\alpha}_{k+1}, \boldsymbol{\nu}_{k+1})| < \epsilon$ the subgradient ascent algorithm is run for an additional $L$ iterations,
and the ergodic primal solutions are computed.
\end{remark}

The averaged power allocations obtained from
$(\bar{\boldsymbol{\mu}}_L, \bar{\boldsymbol{\pi}}_L)$ satisfy the MAC
constraints~\eqref{eq:MAC} but may not correspond to a single SIC decoding
order. To retain SIC-based decoding, we record the empirical frequency of each decoding order over the averaging window and construct a randomized policy by defining scheduling probabilities according to these frequencies. The policy then selects a rate vector and a decoding order, for which the corresponding power allocation is computed via SIC using~\eqref{eq:POW_SIC}.

\begin{remark}
The proposed dual decomposition approach offers substantial computational
advantages over directly solving~\eqref{eq:SRP_pi}, which would require
handling $|\boldsymbol{\mathcal{H}}| \times |\boldsymbol{\mathcal{R}}|
\times (2^M - 1)$ MAC constraints, exponential in the number of users.
By exploiting the SIC structure, the power allocation is obtained
analytically, completely eliminating these constraints. Moreover,
Lemma~\ref{lm:Lemma_decoding} reduces the decoding order optimization from
factorial complexity $\mathcal{O}(M!)$ to $\mathcal{O}(M\log M)$ per
$(\mathbf{h},\boldsymbol{\rho})$ via sorting.

The algorithm computes $\boldsymbol{\mu}$ and $\boldsymbol{\eta}$ in
closed form using~\eqref{eq:mu_optimal} and~\eqref{eq:eta_opt_reg}, avoiding
iterative convex optimization. For each $\mathbf{h}$, evaluating
$A(\mathbf{h},\boldsymbol{\rho})$ over
$\boldsymbol{\rho}\in\boldsymbol{\mathcal{R}}$ incurs complexity
$\mathcal{O}(|\boldsymbol{\mathcal{R}}|\,M)$, while computing
$\boldsymbol{\eta}^*$ requires $\mathcal{O}(M)$. 
\end{remark}

\subsubsection{VA-SRP under TDMA Scheme}
We also solve~\eqref{eq:SRP_pi} under a TDMA constraint, where at most one user
transmits $\rho>0$ bits in any slot. This is enforced by restricting the
rate vector $\boldsymbol{\rho}$ to the set
\begin{align*}
   S \triangleq \{\mathbf{0}\} \;\cup\; \bigcup_{i=1}^{M}
\left\{ \rho\,\mathbf{e}_i \;\middle|\; \rho \in \{1,2,\ldots,r_{\max}\} \right\}, 
\end{align*}
instead of $\boldsymbol{\mathcal{R}}$.
Here, $\mathbf{e}_i$ denotes the $i^{\rm th}$ standard basis vector in $\mathbb{R}^M$,
and $\mathbf{0}$ is the all-zero vector. Thus, $S$ includes either no
transmission or exactly one active user transmitting between $1$ and
$r_{\max}$ bits, ensuring the TDMA constraint.

\subsection{Performance Bound on VA-SRP}

We now establish a provable approximation guarantee for the VA-SRP solution.

\begin{theorem}\label{thm:Bound}
The optimal objective value of~\eqref{eq:SRP_pi},  $V_{\rm SRP}$, is at most twice the optimal objective value of~\eqref{eq:main-opt-problem}, $V_{\rm opt}$,  i.e., $V_{\rm SRP} \leq 2V_{\rm opt}$.
\end{theorem}

\begin{proof}
We follow a standard lower-bounding and policy-construction argument,
adapted from~\cite{Gangadhar_VAoI}. See Appendix~C for details.
\end{proof}

\subsection{VA-SRP with Power Adjustment (VA-SRP w/ PA)} 

The proposed VA-SRP is both VAoI- and queue-agnostic, and may therefore schedule a user even when its queue is empty. While this does not affect VAoI (which is zero when the queue is empty), it leads to unnecessary power accounting in the SRP formulation. In practice, no power is consumed when a scheduled user has an empty queue. Consequently, the actual average power usage is lower than that imposed by the original SRP constraints. We therefore refer to the policy obtained from~\eqref{eq:SRP_pi} as \emph{VA-SRP without Power Adjustment (VA-SRP w/o PA)}. Next, we introduce refined power constraints and denote the corresponding policy as \emph{VA-SRP with Power Adjustment (VA-SRP w/ PA)}.

We define the instantaneous transmit power of User~$i$ in slot $t$ as
$P_i(t) = f_i(\mathbf{h}(t), \boldsymbol{\rho}(t))$ if $Q_i(t)=1$, and
$P_i(t)=0$ otherwise.  The expected per-slot transmit power is 
\begin{align*}
\mathbb{E}[P_i(t)] = \mathbb{E}[f_i(\mathbf{h}(t), \boldsymbol{\rho}(t)) \mid Q_i(t) = 1] \mathbb{P}(Q_i(t) = 1).
\end{align*}
Conditioned on $Q_i(t)=1$, the expected transmit power under VA-SRP is
\begin{align*}
\mathbb{E}[f_i(\mathbf{h}(t), \boldsymbol{\rho}(t)) \mid Q_i(t) = 1] = \sum_{\boldsymbol{\rho} \in \boldsymbol{\mathcal{R}}} \mathbb{E}_{\mathbf{h}}[\mu(\mathbf{h}, \boldsymbol{\rho}) f_i(\mathbf{h}, \boldsymbol{\rho})].
\end{align*}
Using the queue occupancy probability, $\mathbb{P}(Q_i(t) = 1)$, from~\eqref{eq:Q_prob}, the unconditional expected per-slot transmit power, 
\begin{align*}
\mathbb{E}[P_i(t)] = \frac{\lambda_i \sum_{\boldsymbol{\rho} \in \boldsymbol{\mathcal{R}}} \mathbb{E}_{\mathbf{h}}[\mu(\mathbf{h}, \boldsymbol{\rho}) f_i(\mathbf{h}, \boldsymbol{\rho})]}{\lambda_i(1-p_i(\boldsymbol{\mu})) + p_i(\boldsymbol{\mu})}, \; \forall i \in \mathcal{M}.
\end{align*}
Hence, the long-term average power under VA-SRP w/ PA is $P_i(\boldsymbol{\mu}, \boldsymbol{F}) = \mathbb{E}[P_i(t)]$ and the average power constraints are given by $P_i(\boldsymbol{\mu}, \boldsymbol{F}) \leq \bar{P}_i, \; \forall i \in \mathcal{M}$, which can be
equivalently expressed as 
\begin{align}\label{eq:const_with_PA}
&P'_i(\boldsymbol{\mu}, \boldsymbol{F}) 
    = \lambda_i \sum_{\boldsymbol{\rho} \in \boldsymbol{\mathcal{R}}} \mathbb{E}_{\mathbf{h}}[\mu(\mathbf{h}, \boldsymbol{\rho}) f_i(\mathbf{h}, \boldsymbol{\rho})] \nonumber\\
&\leq \bar{P}'_i(\boldsymbol{\mu}) = \bar{P}_i \left(\lambda_i - \lambda_ip_i(\boldsymbol{\mu}) + p_i(\boldsymbol{\mu})\right), \quad\forall i \in \mathcal{M}.
\end{align}
Replacing the power constraint~\eqref{eq:pow_pi} with~\eqref{eq:const_with_PA} in~\eqref{eq:SRP_pi} yields the VA-SRP w/ PA.

 \subsection{Heuristic Policies}
In this section, we describe several heuristic  policies.
\subsubsection{VAoI-Aware Greedy Policy}
In this policy, all transmission vectors $\boldsymbol{\rho}$ consistent with the current queue state are considered (i.e., $\rho_i=0$ if $Q_i=0$). For each $\boldsymbol{\rho}$, transmit powers are obtained by solving instantaneous power minimization problem subject to the MAC constraints. Among all rate vectors, only those that satisfy the running average power and distortion constraints are retained. The optimal transmission vector $\boldsymbol{\rho}^*$ is chosen to minimize the instantaneous total VAoI while maximizing the aggregate number of transmitted bits.

\subsubsection{Max-VAoI-First Policy (TDMA)}
This policy is a TDMA-restricted variant of the VAoI-aware Greedy policy. 
In each slot, at most one user is allowed to transmit, and the user with the
largest instantaneous VAoI is selected for transmission. 
The transmission rate is then chosen to satisfy the running average power and
distortion constraints while maximizing the number of transmitted bits.

\subsubsection{Round-Robin Policy}
This policy is VAoI-agnostic and selects users in a round-robin manner, irrespective of VAoI and queue states. Each user is given an
equal opportunity to be served in a cyclic order, and only the scheduled user is
allowed to transmit in a given time slot. For the selected user, $\rho^* \in \{0,1,\ldots,r_{\rm max}\}$ is chosen to satisfy the running average power and distortion constraints while maximizing the number of transmitted bits.

\section{Numerical Results}
This section presents numerical results. Unless otherwise stated, we use
$g(x)=\log_2(1+x)$ and $\delta(x)=(1-x/r_{\max})^2$.

\input{Plots/age_vs_P_D.tex}

\input{Plots/age_vs_lambda.tex}

\input{Plots/achievable_region.tex}

In Fig.~\ref{fig-Power_Dist}, we present the variation of the average VAoI under different scheduling policies. Fig.~\ref{fig-pow_vs_VAoI} illustrates how the long-term expected average VAoI varies with the bound on the average transmit power for users, $\bar{P}_i$, for both NOMA and TDMA schemes under various policies, including the proposed VA-SRP and several heuristic baselines.

As $\bar{P}_i$ increases, the average VAoI decreases and eventually saturates,
since higher power budgets enable more frequent status update deliveries.
Under NOMA, VA-SRP w/ PA consistently outperforms VA-SRP w/o PA by eliminating
redundant power usage when queues are empty, enabling more frequent transmissions.
VA-SRP w/ PA also outperforms Max-VAoI-First and Round-Robin policies and achieves
performance comparable to the VAoI-Aware Greedy policy.
Importantly, VA-SRP is simpler, as it is VAoI-agnostic and avoids per-slot VAoI tracking.

Now, we compare the TDMA and NOMA schemes. At lower values of $\bar{P}_i$, both TDMA and NOMA schemes exhibit similar performance under the VA-SRP and Greedy policies. As the available transmit power increases, NOMA begins to outperform TDMA. 
At low power levels, NOMA effectively behaves like TDMA by serving a single user,
whereas at higher power levels it enables simultaneous multi-user transmissions.

Under VA-SRP, as $\bar{P}_i$ increases, the average VAoI under NOMA approaches zero,
whereas under TDMA it saturates at $\sum_{i=1}^M w_i \lambda_i (M-1)$ for an $M$-user system.
This behavior follows directly from the status update delivery probability constraints. The long-term expected average VAoI under VA-SRP is given by $\sum_{i = 1}^M w_i \lambda_i (p_i^{-1} - 1)$, where $p_i \in [0,1]$ denotes the probability of successful status update delivery from User~$i$. 
In the TDMA scheme, even with infinite power, only one user can transmit a status update at any given time. Hence, with infinite $\bar{P}_i$, the optimal status update delivery probability becomes $p^*_i = 1/M,~\forall i$, leading to a minimum achievable average VAoI of $ \sum_{i = 1}^M w_i \lambda_i (M - 1)$. In contrast, with infinite $\bar{P}_i$, NOMA enables simultaneous transmissions to multiple users, allowing $p^*_i = 1,~\forall i$, which results in a minimum achievable average VAoI of zero.

In Fig.~\ref{fig-dist_vs_VAoI}, we study the variation of the average VAoI with respect to the bound on the average distortion for users, $\bar{D}_i$. Similar trends as discussed previously are observed from the figure: as $\bar{D}_i$ increases, the average VAoI decreases and eventually saturates. For a given $\bar{P}_i$, at lower values of $\bar{D}_i$, the Greedy policy outperforms all other policies. However, as $\bar{D}_i$ increases, the VA-SRP w/ PA with the NOMA scheme begins to outperform the Greedy policy.

In Fig.~\ref{fig-lambda_vs_VAoI}, we illustrate how the average VAoI varies with the packet arrival probability, $\lambda_i$. As $\lambda_i$ increases, the average VAoI also rises. Nevertheless, the relative performance trends among the different policies remain consistent with those observed in the previous results. It is observed that at $\lambda_i = 1$, with the NOMA scheme, both VA-SRP w/ PA and VA-SRP w/o PA achieve identical performance, since the constraint in \eqref{eq:pow_pi} and \eqref{eq:const_with_PA} becomes equal when $\lambda_i = 1$. However, for $\lambda_i < 1$, the VA-SRP w/ PA outperforms the VA-SRP w/o PA.

In Fig.~\ref{fig-region}, we illustrate the achievable long-term average VAoI regions for the TDMA and NOMA schemes under the VA-SRP w/o PA for the two-user case. The regions are obtained by varying the user weights $w_1$ and $w_2$ between $0$ and $1$, subject to $w_1 + w_2 = 1$. The figure shows that the achievable region of the NOMA scheme entirely encompasses that of the TDMA scheme. In TDMA, achieving a lower average VAoI for one user leads to a significantly higher VAoI for the other user, since it can serve only one user at a time. 

\input{Plots/table.tex}

In Table~\ref{tab-compute_time}, we report the offline and online computational times of the different policies. The VA-SRP takes a higher offline computational time than the VAoI-Aware Greedy, Max-VAoI-First, and Round-Robin policies, as these require no offline optimization. This is because the VA-SRP is computed using a subgradient ascent
algorithm with several hundred iterations (until convergence), where each iteration involves
sorting operations, maximization over a set, and the computation of primal 
variables using closed-form expressions.
Conversely, for online execution, the VAoI-Aware Greedy, Max-VAoI-First and Round-Robin policies require significantly more time than VA-SRP, since they track the VAoI, average power, and distortion of all users and must solve power-minimization problems under the MAC constraints. Overall, VA-SRP achieves the lowest average VAoI among the considered policies at
the cost of higher offline computation, while maintaining low online complexity
due to its stationary and VAoI-agnostic structure.

\input{Plots/distortion_fun_J.tex}

As discussed in the introduction, not all bits within an update are equally important in practice. 
Our framework captures such bit-priority structures through a distortion metric $\delta(\rho)$, 
which depends on the number of transmitted bits, $\rho$. 
This allows modeling scenarios where the first few bits carry task-critical information, 
while the remaining bits provide refinement.
To illustrate this flexibility, Fig.~\ref{fig:delta_fun} presents different distortion functions 
corresponding to different application-specific bit-priority structures. 
For different distortion functions, we show the variation of the long-term average VAoI ($\bar{\Delta}$) 
and scheduling probabilities $\mu(h,\rho)$ under the VA-SRP, for fixed bound on average transmit power, $\bar{P}$.

From Fig.~\ref{fig:delta_fun}, we observe that for a fixed distortion function 
(e.g., $\delta_1(\rho)$), relaxing the distortion constraint (i.e., increasing $\bar{D}$) 
leads to a reduction in the average VAoI. This is achieved by transmitting fewer bits more frequently, 
as reflected by increased scheduling probabilities for smaller $\rho$ and reduced probabilities for larger $\rho$. 
Conversely, tighter distortion constraints favor transmitting more bits per update but less frequently,
as required by applications demanding near-complete packet reconstruction.
Importantly, the choice of the distortion function significantly affects the scheduling policy 
even for a fixed $\bar{D}$. For example, with a convex distortion function $\delta_1(\rho)$, 
VA-SRP tends to transmit fewer bits with higher probability, resulting in a
lower average VAoI. In contrast, step and linear distortion functions
$\delta_2(\rho)$ and $\delta_3(\rho)$ favor transmitting more bits less frequently,
resulting in higher average VAoI. This shows that distortion functions can be selected based on application requirements 
regarding the preferred number of transmitting bits.
Overall, by appropriately selecting $\delta(\rho)$, the proposed framework enables a flexible trade-off 
between average VAoI and the frequency of transmitting a larger number of bits. 
This allows the VA-SRP to adapt across applications ranging from 
semantic or flag-based updates to full-packet reconstruction scenarios.

\section{Conclusion}
This paper studied VAoI minimization in uplink NOMA systems under average power and information-quality constraints captured through a general distortion metric. We developed a VAoI-agnostic stationary randomized policy that jointly optimizes scheduling, bit allocation, and power control, and achieves a provable 2-approximation to the globally optimal VAoI. By leveraging Lagrangian dual decomposition, we reduce the computational complexity of obtaining the policy. The algorithm has polynomial-time complexity per iteration, dominated by a sorting step to determine the decoding order for SIC.
Numerical results demonstrate that NOMA substantially outperforms TDMA, achieving near-zero VAoI at high power budgets through simultaneous multi-user transmissions. At the same time, TDMA remains limited by its single-user-per-slot constraint. Our distortion framework accommodates diverse bit-priority structures, from uniform importance to metadata-critical scenarios, enabling application-specific trade-offs between timeliness and information quality. Future work includes extending the framework to incorporate bit buffering across slots, time-varying channel statistics, and energy harvesting constraints.

\bibliographystyle{IEEEtran}
\bibliography{references}

\appendix
\subsection{Proof of Theorem~\ref{thm:convexity}}
We prove that~\eqref{eq:SRP_pi} is a convex problem by showing that
the objective is convex and all constraints define convex sets.

\textit{Objective function~\eqref{eq:obj_pi}:}
Recall that
$p_i(\boldsymbol{\mu}) = \sum_{\boldsymbol{\rho}} \sum_{\mathbf{h}}
\mu(\mathbf{h}, \boldsymbol{\rho}) \mathbb{I}\{\rho_i > 0\} \mathbb{P}(\mathbf{h})$,
which is affine in $\boldsymbol{\mu}$.
Since $f(x)=1/x$ is convex and decreasing for $x>0$, the composition
$1/p_i(\boldsymbol{\mu})$ is convex~\cite{Boyd}.
Hence, \eqref{eq:obj_pi}, being a nonnegative weighted sum of convex
functions, is convex.

\textit{Power constraint~\eqref{eq:pow_pi}:}
This constraint is linear in $\boldsymbol{\pi}$ and independent of
$\boldsymbol{\mu}$, and therefore defines a convex set.

\textit{Distortion constraint~\eqref{eq:dist_pi}:}
The left-hand side
$D'_i(\boldsymbol{\mu}) = \lambda_i \sum_{\boldsymbol{\rho}}
\mathbb{E}_{\mathbf{h}}[\mu(\mathbf{h}, \boldsymbol{\rho})
\delta(\rho_i)\mathbb{I}\{\rho_i>0\}]$
is affine in $\boldsymbol{\mu}$.
The right-hand side
$\bar{D}'_i(\boldsymbol{\mu}) =
\bar{D}_i\big(\lambda_i - \lambda_i p_i(\boldsymbol{\mu}) + p_i(\boldsymbol{\mu})\big)$
is also affine since $p_i(\boldsymbol{\mu})$ is affine.
Thus, the constraint defines a convex set.

\textit{MAC constraint~\eqref{eq:MAC_pi}:}
For fixed $(\mathcal{S},\mathbf{h},\boldsymbol{\rho})$, the constraint
\[
\mu(\mathbf{h}, \boldsymbol{\rho}) g^{-1}\!\Big(\sum_{i\in\mathcal{S}}\rho_i\Big)
\le \sum_{i\in\mathcal{S}} \pi_i(\mathbf{h}, \boldsymbol{\rho}) h_i
\]
is linear in $(\boldsymbol{\mu},\boldsymbol{\pi})$, since
$g^{-1}(\sum_{i\in\mathcal{S}}\rho_i)$ is constant.
Hence, each constraint defines a halfspace, whose intersection is convex.

\textit{Constraints~\eqref{eq:mu_bound} and~\eqref{eq:mu_sum}:}
These are affine constraints in $\boldsymbol{\mu}$ and therefore convex.

Since the objective function is convex and all constraints define convex sets,
problem~\eqref{eq:SRP_pi} is a convex optimization problem.

\balance
\subsection{Proof of Lemma~\ref{lm:Lemma_decoding}}

We first prove the lemma under the assumption that $\rho_{\theta_j} > 0$ for all $j \in \mathcal{M}$, and then extend the result to the case where $\rho_{\theta_j} = 0$ for some users.

\paragraph{Case 1 ($\rho_{\theta_j} > 0$ for all $j \in \mathcal{M}$)}
For a given decoding order $\boldsymbol{\theta}$, the objective function in
\eqref{eq:decoding_opt} can be written as
\begin{align*}
P(\boldsymbol{\theta})
= \sum_{j=1}^{M}
\frac{\beta_{\theta_j} g^{-1}(\rho_{\theta_j})}{h_{\theta_j}}
\prod_{k=j+1}^{M} \left(1 + g^{-1}(\rho_{\theta_k})\right),
\end{align*}
where all terms are strictly positive under the assumption
$\rho_{\theta_j} > 0$.

Consider any two adjacent users $x=\theta_t$ and $y=\theta_{t+1}$ in the
decoding order, and let $\boldsymbol{\theta}'$ denote the order obtained
by swapping $x$ and $y$, with all other user positions
unchanged. All terms in $P(\boldsymbol{\theta})$ and
$P(\boldsymbol{\theta}')$ are identical except those involving users
$x$ and $y$. A direct comparison yields
\begin{align*}
P(\boldsymbol{\theta}) - P(\boldsymbol{\theta}')
= G \, g^{-1}(\rho_x) g^{-1}(\rho_y)
\left(\frac{\beta_x}{h_x} - \frac{\beta_y}{h_y}\right),
\end{align*}
where $G>0$ is a common positive factor independent of $x$ and $y$.

Hence, $P(\boldsymbol{\theta}) \le P(\boldsymbol{\theta}')$ if and only if
$\beta_x/h_x \le \beta_y/h_y$. Therefore, any decoding order containing
an adjacent pair violating this condition can be improved by swapping the
pair. By repeated adjacent exchanges, the optimal decoding order must
satisfy
\begin{align*}
\frac{\beta_{\theta_1^*}}{h_{\theta_1^*}}
\le \frac{\beta_{\theta_2^*}}{h_{\theta_2^*}}
\le \cdots
\le \frac{\beta_{\theta_M^*}}{h_{\theta_M^*}},
\end{align*}
or equivalently,
\(
h_{\theta_1^*}/\beta_{\theta_1^*}
\ge \cdots \ge
h_{\theta_M^*}/\beta_{\theta_M^*}.
\)

\paragraph{Case 2 ($\rho_{\theta_j}=0$ for some users)}

If $\rho_i=0$, then $g^{-1}(\rho_i)=0$, and User~$i$ contributes zero power and zero cost to $P(\boldsymbol{\theta})$, regardless of its position in the decoding order. However, placing such users early in the SIC order is undesirable, since earlier decoding stages correspond to higher effective interference levels and require higher transmit powers.

This is resolved by modifying the sorting metric to
$h_i \mathbb{I}\{\rho_i>0\}/\beta_i$, which preserves the optimal ordering
among transmitting users while forcing all non-transmitting users to the
end of the decoding order. Their relative order is immaterial since they
incur zero power.

Accordingly, the optimal decoding order satisfies
\begin{align*}
\frac{h_{\theta_1^*}\mathbb{I}\{\rho_{\theta_1^*}>0\}}{\beta_{\theta_1^*}}
\ge \cdots \ge
\frac{h_{\theta_M^*}\mathbb{I}\{\rho_{\theta_M^*}>0\}}{\beta_{\theta_M^*}},
\end{align*}
which completes the proof.

\subsection{Proof of Theorem~\ref{thm:Bound}}
We provide a complete proof of Theorem~\ref{thm:Bound}, following a three-step argument.

\textit{Step 1: Lower bound on optimal VAoI.}
For a single user, the expected sum of VAoI over an inter-delivery interval of $I$ slots is $\mathbb{E}[\Delta_I] = \lambda (I^2 - I)/2$~\cite{Gangadhar_VAoI}, where $\lambda$ is the packet arrival probability. Consider a time horizon of $T$ slots. Let $D(T)=\sum_{t=1}^T \mathbb{I}\{\rho(t)>0\}$ denote the number of successful updates, $I[1],\ldots,I[D(T)]$ denotes inter-delivery intervals, and $R$ denote the remaining slots after the last delivery. The time-average expected VAoI is expressed as
\begin{align*}
    \frac{1}{T}\sum_{t=1}^T \mathbb{E}[\Delta(t)] = \frac{\lambda}{2}\left(\sum_{k=1}^{D(T)} \frac{I^2[k] - I[k]}{T} + \frac{R^2-R}{T}\right).
\end{align*}
By Jensen's inequality ($\bar{\mathbb{M}}[I^2] \geq (\bar{\mathbb{M}}[I])^2$, where $\bar{\mathbb{M}}[\cdot]$ denotes the sample mean) and minimizing with respect to $R$ yields lower bound to the long-term average VAoI as
\begin{align*}
    \lim_{T \to \infty} \frac{1}{T} \sum_{t=1}^T \mathbb{E}[\Delta(t)]  \geq \frac{\lambda}{2}\left(\frac{1}{q} - 1\right),
\end{align*}
where $q = \lim_{T \to \infty} (1/T)\mathbb{E}[D(t)] =   \lim_{T \to \infty} (1/T)\mathbb{E}[ \sum_{t=1}^T\mathbb{I}\{\rho(t) > 0 \}]$ is  the probability of status update delivery. Recall that $\rho>0$ corresponds to a successful update. For the $M$-user system, we have
\begin{align*}
    V_{\phi} = \lim_{T \to \infty} \frac{1}{T} \sum_{t=1}^T \sum_{i=1}^M w_i\mathbb{E}[\Delta_i(t)] \geq \frac{1}{2}\sum_{i=1}^M w_i \lambda_i\left(\frac{1}{q^{\phi}_i} - 1\right),
\end{align*}
for any policy $\phi$, where $q^{\phi}_i$ is the probability of status update delivery for User~$i$ under policy $\phi$. Minimizing over all admissible policies yields $V_{\rm opt}\ge L_B$, where 
\begin{align*}
    V_{\rm opt}=\min_{\phi}\; &V_{\phi},\\
    \text{subject to}\;\; &\eqref{eq:MAC}, \eqref{eq:age_evolution}, \eqref{eq:pow_constraint}, \eqref{eq:dist_constraint},
\end{align*}
and
\begin{align}\label{eq:LB_problem}
L_B = \min_{\phi} \;\; &
\frac{1}{2} \sum_{i=1}^{M} w_i \lambda_i 
\left(\frac{1}{q^{\phi}_i} - 1\right),\\
\text{subject to} \;\; &
\eqref{eq:MAC},\; \eqref{eq:age_evolution},\;
\eqref{eq:pow_constraint},\; \eqref{eq:dist_constraint}.\nonumber
\end{align}

\textit{Step 2: Construction of a feasible VA-SRP.}
Let $\phi^{\rm LB}$ denote the optimal policy achieving $L_B$, which specifies transmission decisions $(\rho_i^{\rm LB}(t), f^{\rm LB}_i(t))_{i \in \mathcal{M}}$ for each time slot $t$.
We construct a VA-SRP by defining scheduling probabilities based on the empirical frequencies under $\phi^{\rm LB}$:
\begin{align*}
\hat{\mu}(\mathbf{h}, \boldsymbol{\rho}) = \lim_{T \to \infty} \frac{1}{|\mathcal{T}_T(\mathbf{h})|} \mathbb{E}\left[\sum_{t \in \mathcal{T}_T(\mathbf{h})} \mathbb{I}\{\boldsymbol{\rho}^{\rm LB}(t) = \boldsymbol{\rho}\}\right],
\end{align*}
for all $\mathbf{h} \in \boldsymbol{\mathcal{H}}, \boldsymbol{\rho} \in \boldsymbol{\mathcal{R}}$,
where $\mathcal{T}_T(\mathbf{h}) = \{t \in \{1,\ldots,T\} : \mathbf{h}(t) = \mathbf{h}\}$, and $\boldsymbol{\rho}^{\rm LB}(t) \triangleq (\rho^{\rm LB}_i(t))_{i\in \mathcal{M}}$. The expectation is taken over the randomness in the system under the lower-bound policy $\phi^{\rm LB}$.

We have $\hat{\boldsymbol{\mu}} \triangleq (\hat{\mu}(\mathbf{h}, \boldsymbol{\rho}))_{\mathbf{h} \in \boldsymbol{\mathcal{H}}, \boldsymbol{\rho} \in \boldsymbol{\mathcal{R}}}$, and the probability of status update delivery for User $i$ under constructed VA-SRP, $p_i(\hat{\boldsymbol{\mu}}) = \sum_{\boldsymbol{\rho} \in \boldsymbol{\mathcal{R}}} \mathbb{E}_{\mathbf{h}}[\hat{\mu}(\mathbf{h}, \boldsymbol{\rho}) \, \mathbb{I}\{\rho_i > 0\}]$. Let $\hat{V}_{\rm SRP}$ denote the long-term expected average VAoI achieved by
the constructed VA-SRP. Using \eqref{eq:SRP_VAoI}, we obtain $\hat{V}_{\rm SRP} = \sum_{i =1}^M w_i \lambda_i (1/p_i(\hat{\boldsymbol{\mu}}) - 1)$, with $\hat{\boldsymbol{\mu}}$ satisfying average power and distortion constraints.

We define the probability of status update delivery under $\phi^{\rm LB}$ as
\begin{align*}
    q_i^{\rm LB} &= \lim_{T \to \infty} \frac{1}{T} \mathbb{E}\left[\sum_{t=1}^T \mathbb{I}\{\rho_i^{\rm LB}(t) > 0\}\right] \\
    &= \sum_{\mathbf{h} \in \boldsymbol{\mathcal{H}}} \sum_{\boldsymbol{\rho} \in \boldsymbol{\mathcal{R}}} \mathbb{P}(\mathbf{h}) \hat{\mu}(\mathbf{h}, \boldsymbol{\rho}) \mathbb{I}\{\rho_i > 0\}, \; \forall i \in \mathcal{M}.
\end{align*}
Substituting $q_i^{\rm LB}$ into the lower-bound expression yields
\begin{align}\label{eq:LB_opt}
    L_B = \frac{1}{2} \sum_{i = 1}^M w_i \lambda_i \left(\frac{1}{q_i^{\rm LB}} - 1 \right).
\end{align}
Since the constructed VA-SRP employs the same scheduling probabilities
$\hat{\mu}(\mathbf{h},\boldsymbol{\rho})$ for all 
$(\mathbf{h},\boldsymbol{\rho})$, the resulting status update probability for User~$i$ satisfies $
p_i(\hat{\boldsymbol{\mu}}) = q_i^{\rm LB}$,
and the corresponding average VAoI is given by
$\hat{V}_{\rm SRP} = \sum_{i=1}^M w_i \lambda_i \left( {1}/{q_i^{\rm LB}} - 1 \right)$.

Since the constructed VA-SRP inherits the scheduling frequencies of $\phi^{\rm LB}$,
it satisfies the average power and distortion constraints. Comparing with
\eqref{eq:LB_opt} yields $\hat{V}_{\rm SRP}=2L_B$.

\emph{Power constraint satisfaction:} 
Power constraint satisfaction follows directly from~\cite{Gangadhar_VAoI}:
the constructed VA-SRP inherits feasibility since $\hat{\mu}(\mathbf{h},\boldsymbol{\rho})$
is derived from the lower-bound policy.

\emph{Distortion constraint satisfaction:} 
Since the lower-bound problem \eqref{eq:LB_problem} includes distortion constraints, $\phi^{\rm LB}$ satisfies the average distortion constraints:
\begin{align}\label{eq:LB_distortion}
    \lim_{T \to \infty} \frac{1}{T} \mathbb{E}\left[\sum_{t = 1}^T d_i(t)\right] \leq \bar{D}_i, \quad \forall i \in \mathcal{M}.
\end{align}
 We can decompose the left-hand side by conditioning on the queue state. Under any policy, the per-slot expected distortion is: $\mathbb{E}[d_i(t)] = \mathbb{E}[d_i(t) | Q_i(t)=1] \cdot \mathbb{P}(Q_i(t)=1)$,
since $d_i(t) = 0$ when $Q_i(t) = 0$.
For the lower-bound policy, the conditional expected distortion is
\begin{align*}
     \mathbb{E}^{\phi^{\rm LB}}&[d_i(t) | Q_i(t) = 1]\\ 
     &=\lim_{T \to \infty} \frac{1}{T} \mathbb{E}\left[\sum_{t=1}^T \delta(\rho_i^{\rm LB}(t)) \mathbb{I}\{\rho_i^{\rm LB}(t) > 0\}\right]\\
     &= \sum_{\mathbf{h}\in \boldsymbol{\mathcal{H}}} \sum_{\boldsymbol{\rho}\in \boldsymbol{\mathcal{R}}} \mathbb{P}(\mathbf{h}) \hat{\mu}(\mathbf{h}, \boldsymbol{\rho}) \delta(\rho_i) \mathbb{I}\{\rho_i \smash{>} 0\}.
\end{align*}
The constructed SRP uses the same scheduling probabilities $\hat{\mu}(\mathbf{h}, \boldsymbol{\rho})$ for all $(\mathbf{h}, \boldsymbol{\rho})$. Therefore, the conditional expected distortion under the constructed VA-SRP is
\begin{align*}
    \mathbb{E}^{\rm SRP}[d_i(t) | Q_i(t)=1] = \mathbb{E}^{\phi^{\rm LB}}[d_i(t) | Q_i(t)=1].
\end{align*}
From~\eqref{eq:Q_prob}, the steady-state queue occupancy probability depends only on $\lambda_i$ and the status update probability $p_i$.
Since $p_i(\hat{\boldsymbol{\mu}}) = q_i^{\rm LB}$, we have
\begin{align*}
    \mathbb{P}^{\rm SRP}(Q_i(t)=1) = \frac{\lambda_i}{\lambda_i(1-q_i^{\rm LB}) + q_i^{\rm LB}} = \mathbb{P}^{\phi^{\rm LB}}(Q_i(t)=1).
\end{align*}
Therefore, the expected average distortion under the constructed SRP is
\begin{align*}
    \mathbb{E}^{\rm SRP}[d_i(t)] &= \mathbb{E}^{\rm SRP}[d_i(t) | Q_i(t)=1] \cdot \mathbb{P}^{\rm SRP}(Q_i(t)=1) \\
    &= \mathbb{E}^{\phi^{\rm LB}}[d_i(t) | Q_i(t)=1] \cdot \mathbb{P}^{\phi^{\rm LB}}(Q_i(t)=1) \\
    &= \mathbb{E}^{\phi^{\rm LB}}[d_i(t)] \leq \bar{D}_i,
\end{align*}
where the last inequality follows from~\eqref{eq:LB_distortion}. Hence, the constructed VA-SRP satisfies the distortion constraints.

\textit{Step 3: Establishing the 2-approximation.}
Since the scheduling probabilities $\hat{\boldsymbol{\mu}}$ of the constructed
VA-SRP constitute a feasible solution to problem~\eqref{eq:SRP_pi}, the
corresponding objective value satisfies $\hat{V}_{\rm SRP} \geq V_{\rm SRP}$.
Combining all the established bounds, we obtain
\begin{align*}
 L_B \leq V_{\rm opt} \leq V_{\rm SRP} \leq \hat{V}_{\rm SRP} = 2L_B \leq 2V_{\rm opt}.
\end{align*}
Consequently, the VA-SRP achieves a $2$-approximation of the optimal VAoI, i.e.,  $V_{\rm SRP} \leq 2V_{\rm opt}$.

\end{document}

%% file: Plots/age_vs_P_D.tex
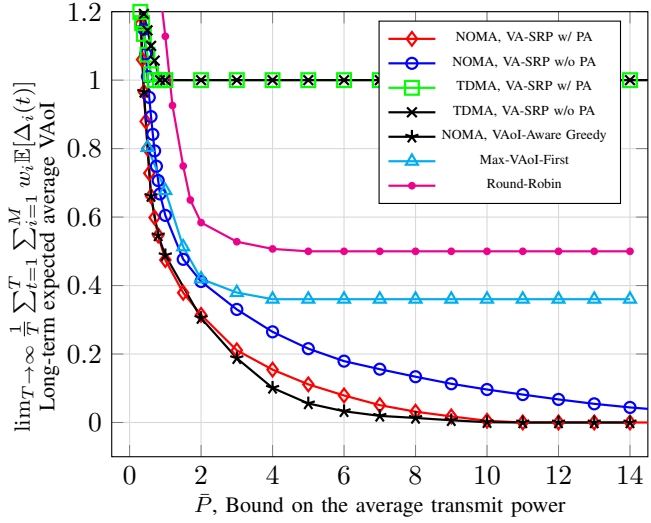
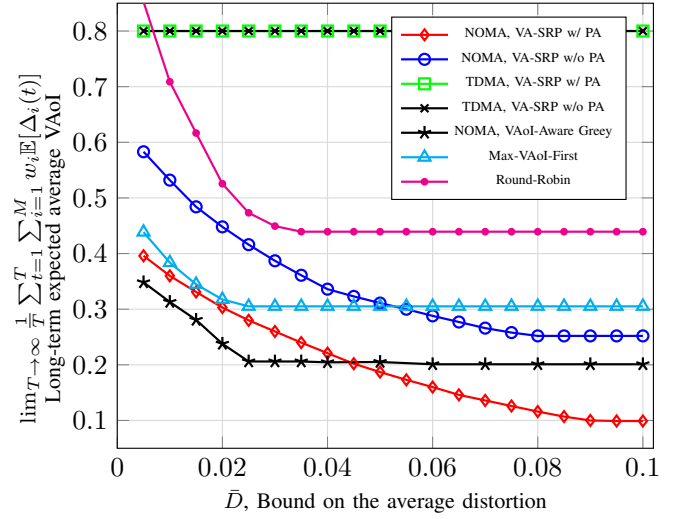
\begin{figure*}[t]
    \centering
    \begin{subfigure}[b]{0.48\textwidth}
        \centering
        \begin{tikzpicture}[scale=1.04]
          \begin{axis}[
            grid=both,
       major grid style={draw=gray!30},   
       minor grid style={draw=gray!15}, 
            xtick={0,2,...,20},
            ytick={0,0.2,...,1.6},
            xlabel={$\bar{P}$, Bound on the average transmit power},
            ylabel={$\lim_{T\rightarrow \infty} \frac{1}{T}\sum_{t=1}^{T}\sum_{i=1}^{M}w_i\mathbb{E}[\Delta_i(t)]$\\Long-term expected average VAoI},
           xlabel style={
            at={(axis description cs:0.5,0.03)},
            anchor=north,
        },
        ylabel style={
            at={(axis description cs:0.1,0.45)},
            anchor=south,},
            ymax=1.2,
            ymin=-0.1,
            xmin = -0.5,
            xmax = 14.5,
            every axis label/.append style={align=center, font=\footnotesize},
            legend style={font=\tiny, at={(0.95, 0.98)}}]

        \pgfplotstableread[col sep=comma]{Results/p_bar_vs_age.csv}\datatable
        
          \addplot [color= red, mark=diamond, mark options={solid, red}, mark size = 2.5pt, line width=0.8pt]table[x=PN_PA,y=N_PA]{\datatable};
          \addlegendentry{NOMA, VA-SRP w/ PA};
    
          \addplot [color= blue, mark=o, mark options={solid, blue}, mark size = 2pt, line width=0.8pt]table[x=P2,y=N2]{\datatable};
          \addlegendentry{NOMA, VA-SRP w/o PA};

          \addplot [color= green, mark=square, mark options={green}, mark size = 2.5pt, line width=0.8pt]table[x=PT_PA,y=T_PA]{\datatable};
          \addlegendentry{TDMA, VA-SRP w/ PA};

          \addplot [color= black, mark=x, mark options={black}, mark size = 2.5pt, line width=0.8pt]table[x=P3,y=N3]{\datatable};
          \addlegendentry{TDMA, VA-SRP w/o PA};

          \addplot [color= black, mark=star, mark options={black}, mark size = 2.5pt, line width=0.8pt]table[x=Pow_G,y=Greedy]{\datatable};
          \addlegendentry{NOMA, VAoI-Aware Greedy};

          \addplot [color= cyan, mark=triangle, mark options={cyan}, mark size = 2.5pt, line width=0.8pt]table[x=pow_max,y=Max_aoi2]{\datatable};
          \addlegendentry{Max-VAoI-First};

          \addplot [color= magenta, mark=*, mark options={magenta}, mark size = 1pt, line width=0.8pt]table[x=pow_rr,y=RR]{\datatable};
          \addlegendentry{Round-Robin};
    
          \end{axis}
        \end{tikzpicture}
        \caption{}
        \label{fig-pow_vs_VAoI}
    \end{subfigure}
    \hfill
    \begin{subfigure}[b]{0.48\textwidth}
        \centering
        \begin{tikzpicture}[scale=1.04]
      \begin{axis}[
       grid=both,
       major grid style={draw=gray!30},   
       minor grid style={draw=gray!15},   
       xtick = {0, 0.02,...,0.1},
       ytick = {0, 0.1,..., 0.8},
        xlabel={ $\bar{D}$, Bound on the average distortion},
        ylabel={$\lim_{T\rightarrow \infty} \frac{1}{T}\sum_{t=1}^{T}\sum_{i=1}^{M}w_i\mathbb{E}[\Delta_i(t)]$\\Long-term expected average VAoI},
        xlabel style={
            at={(axis description cs:0.5,0.02)},
            anchor=north,
        },
        ylabel style={
            at={(axis description cs:0.1,0.45)},
            anchor=south,},
        xticklabel style={
        /pgf/number format/fixed,
        /pgf/number format/precision=2
    },
        ymax=0.85,
        ymin= 0.05,
        xmax = 0.102,
        xmin = 0,
        every axis label/.append style={align=center, font=\footnotesize},
        legend style={font=\tiny, at={(0.95,0.97)}}]

    \pgfplotstableread[col sep=comma]{Results/d_bar_vs_age.csv}\datatable
    
          \addplot [color= red, mark=diamond, mark options={solid, red}, mark size = 2pt, line width=0.8pt]table[x=d_bar,y=NOMA_SRP_PA]{\datatable};
          \addlegendentry{NOMA, VA-SRP w/ PA};

          \addplot [color= blue, mark=o, mark options={solid, blue}, mark size = 2pt, line width=0.8pt]table[x=d_bar,y=NOMA_SRP]{\datatable};
          \addlegendentry{NOMA, VA-SRP w/o PA};

          \addplot [color= green, mark=square, mark options={solid, green}, mark size = 2pt, line width=0.8pt]table[x=d_bar,y=TDMA_SRP]{\datatable};
          \addlegendentry{TDMA, VA-SRP w/ PA};

          \addplot [color= black, mark=x, mark options={solid, black}, mark size = 2pt, line width=0.8pt]table[x=d_bar,y=TDMA_SRP]{\datatable};
          \addlegendentry{TDMA, VA-SRP w/o PA};

          \addplot [color= black, mark=star, mark options={solid, black}, mark size = 2.5pt, line width=0.8pt]table[x=d_bar2,y=Greedy]{\datatable};
          \addlegendentry{NOMA, VAoI-Aware Greey};

          \addplot [color= cyan, mark=triangle, mark options={solid, cyan}, mark size = 2.5pt, line width=0.8pt]table[x=d_bar,y=Max_aoi3]{\datatable};
          \addlegendentry{Max-VAoI-First};

          \addplot [color= magenta, mark=*, mark options={solid, magenta}, mark size = 1pt, line width=0.8pt]table[x=d_bar,y=RR]{\datatable};
          \addlegendentry{Round-Robin};

      \end{axis}
    \end{tikzpicture}
    \caption{}
    \label{fig-dist_vs_VAoI}
    \end{subfigure}
\caption{Long-term expected average VAoI vs. (a) power bound $\bar{P}_i = \bar{P}$,
for $\bar{D}_i = 0.05$, $\lambda_i = 0.5$; (b) distortion bound 
$\bar{D}_i = \bar{D}$, for $\bar{P}_i = 2$, $\lambda_i = 0.4, \;\forall i$. 
System parameters: $M=3$, $w_i = 1/3$, $\rho_i \in \{0,1,2\}$, 
$\mathcal{H}_i = \{0.1, 1\}$ with equal probability, $\forall i \in \{1,2,3\}$.}
\label{fig-Power_Dist}
\end{figure*}

%% file: Plots/age_vs_lambda.tex
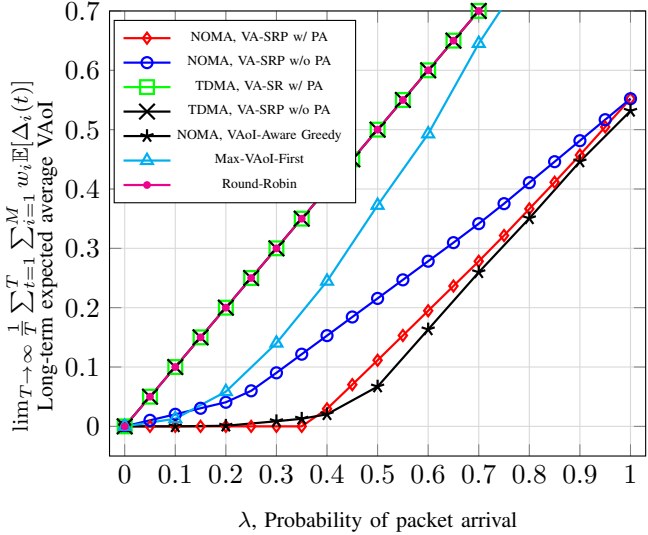
\begin{figure}[t]
    \centering
    \begin{tikzpicture}[scale=1.04]
          \begin{axis}[
            grid=both,
       major grid style={draw=gray!30},   
       minor grid style={draw=gray!15}, 
       xtick = {0, 0.1,...,1},
       ytick = {0, 0.1,...,1},
            xlabel={$\lambda$, Probability of packet arrival},
            ylabel={$\lim_{T\rightarrow \infty} \frac{1}{T}\sum_{t=1}^{T}\sum_{i=1}^{M}w_i\mathbb{E}[\Delta_i(t)]$\\Long-term expected average VAoI},
            xlabel style={
            at={(axis description cs:0.5,-0.01)},
            anchor=north,
        },
        ylabel style={
            at={(axis description cs:0.1,0.45)},
            anchor=south,},
            ymax=0.7,
            ymin=-0.05,
            xmin = -0.03,
            xmax = 1.03,
            every axis label/.append style={align=center, font=\footnotesize},
            legend style={font=\tiny, at={(0.46,0.98)}}]

            \pgfplotstableread[col sep=comma]{Results/lambda_vs_age.csv}\datatable
            
          \addplot [color= red, mark=diamond, mark options={solid, red}, mark size = 2pt, line width=0.8pt]table[x=lambda,y=NOMA_SRP_PA]{\datatable};
          \addlegendentry{NOMA, VA-SRP w/ PA};

          \addplot [color= blue, mark=o, mark options={solid, blue}, mark size = 2pt, line width=0.8pt]table[x=lambda,y=NOMA_SRP]{\datatable};
          \addlegendentry{NOMA, VA-SRP w/o PA};

          \addplot [color= green, mark=square, mark options={green}, mark size = 2.5pt, line width=0.8pt]table[x=lambda,y=TDMA_SRP]{\datatable};
          \addlegendentry{TDMA, VA-SR w/ PA};

          \addplot [color= black, mark=x, mark options={black}, mark size = 4pt, line width=0.8pt]table[x=lambda,y=TDMA_SRP]{\datatable};
          \addlegendentry{TDMA, VA-SRP w/o PA};

          \addplot [color= black, mark=star, mark options={black}, mark size = 2.5pt, line width=0.8pt]table[x=lambda_3,y=Greedy3]{\datatable};
          \addlegendentry{NOMA, VAoI-Aware Greedy};

          \addplot [color= cyan, mark=triangle, mark options={cyan}, mark size = 2.5pt, line width=0.8pt]table[x=lambda_2,y=Max_aoi2]{\datatable};
          \addlegendentry{Max-VAoI-First};

          \addplot [color= magenta, mark=*, mark options={magenta}, mark size = 1pt, line width=0.8pt]table[x=lambda,y=TDMA_SRP]{\datatable};
          \addlegendentry{Round-Robin};

          \end{axis}
        \end{tikzpicture}
\caption{Long-term expected average VAoI vs. probability of packet arrival, $\lambda_i = \lambda$, 
for different scheduling policies. System parameters: $M=3$, $\bar{P}_i = 5$, 
$\bar{D}_i = 0.05$, $w_i = 1/3$, $\rho_i \in \{0,1,2\}$, 
$\mathcal{H}_i = \{0.1, 1\}$ with equal probability, $\forall i \in \{1,2,3\}$.}
        \label{fig-lambda_vs_VAoI}
\end{figure}

%% file: Plots/achievable_region.tex
\begin{figure}[t]
    \centering
    \begin{tikzpicture}[scale=1.04]
      \begin{axis}[
        xlabel={$\lim_{T\rightarrow \infty} \frac{1}{T}\sum_{t=1}^{T} \mathbb{E}[\Delta_1(t)]$, \\Long-term average VAoI of User $1$},
        ylabel={$\lim_{T\rightarrow \infty} \frac{1}{T}\sum_{t=1}^{T} \mathbb{E}[\Delta_2(t)]$,\\ Long-term average VAoI of User $2$},
        xlabel style={
            at={(axis description cs:0.5,0.02)},
            anchor=north, font = \small},
        ylabel style={
            at={(axis description cs:0.1,0.45)},
            anchor=south, font = \small},
        every axis label/.append style={align=center, font=\footnotesize},
         legend style={font=\footnotesize, at={(0.9,0.9)}},
        ymin=0,
        ymax=2.5,
        xmin=0,
        xmax=2.5]
    
       \pgfplotstableread[col sep=comma]{Results/achiev_region_SRP.csv}\datatable 
       
      \addplot [name path=noma, color=blue, very thick,  mark options={blue}, mark size = 1.8pt]table[x=noma_u1,y=noma_u2]{\datatable}; 
       \addlegendentry{NOMA };

      \addplot [dashed, name path=tdma, color=red, ultra thick, mark options={red, dashed}, mark size = 1.8pt]table[x=tdma_u1,y=tdma_u2]{\datatable}; 
       \addlegendentry{TDMA};

      \addplot [name path=out_curve, color=black, mark options={red}, mark size = 1.8pt]table[x=test1,y=test2]{\datatable}; 
    
       \addplot [blue!10] fill between [of = tdma and noma];
       \addplot [red!5, pattern=north east lines, pattern color=red] fill between [of = out_curve and tdma];
  
       \end{axis}
    \end{tikzpicture}
\caption{Achievable long-term expected average VAoI regions for NOMA and TDMA schemes 
under VA-SRP w/o PA by varying user weights $(w_1, w_2)$ subject to 
$w_1 + w_2 = 1$. System parameters: $M=2$, $\bar{P}_i = 5$, $\bar{D}_i = 0.05$, 
$\lambda_i = 0.9$, $\rho_i \in \{0,1,2\}$, $\mathcal{H}_i = \{0.1, 1\}$ 
with equal probability, $\forall i \in \{1,2\}$.}
    \label{fig-region}
\end{figure}
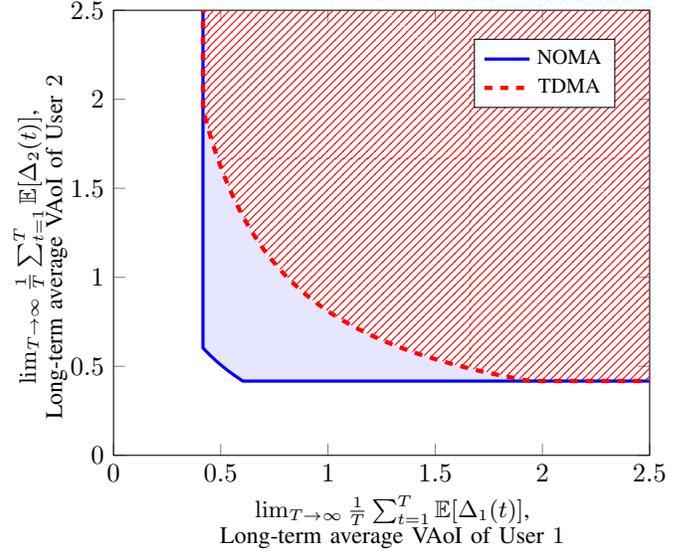

%% file: Plots/table.tex
\begin{table}[t]										
	\centering			
	\caption{ Comparison of average online run-time (per time slot) of different scheduling policies. Each value is obtained by averaging $10$ sample paths, each consisting of $500{\rm k}$ time slots. All the simulations are conducted on a $12^{\rm th}$ Generation Intel Core $i5$-$12600$ Processor with $16$\,GB RAM, using Python~$3.10$. System parameters:  $M = 3$, $\bar{P}_i = 2, \bar{D}_i = 0.06$,  $\lambda_i = 0.5$,  $w_i = 1/3$, $\rho_i \in \{0, 1, 2\}$, $\mathcal{H}_i = \{0.1, 1\}$ with equal probability, $\forall i \in \{1, 2, 3\}$.}				 	\label{tab-compute_time}												
	\begin{tabular}{|p{1.3cm}|p{1cm}|p{1cm}| p{0.9cm}|p{0.8cm}|p{0.8cm}|}									
		\hline  & \vspace{0.03cm} NOMA, VA-SRP w/ PA & \vspace{0.03cm} TDMA, VA-SRP w/ PA & \vspace{0.01cm} NOMA, VAoI-Aware Greedy & \vspace{0.03cm}  Max-VAoI-First & \vspace{0.03cm} Round-Robin   \\								
		
       \hline								
	      Online Run-time (milli seconds) &  \vspace{0.05cm} $0.026$ & \vspace{0.05cm} $0.021$ & \vspace{0.05cm} $176.26$ & \vspace{0.05cm} $51.48$ & \vspace{0.05cm} $50.05$ \\ 
        \hline								
	      Offline Run-time (seconds)  & \vspace{0.05cm} $85.7$ & \vspace{0.05cm} $16.3$ & \vspace{0.05cm} - & \vspace{0.05cm} - & \vspace{0.05cm} - \\ 
       \hline
        Average VAoI & \vspace{0.02cm}  $0.2762$ & \vspace{0.02cm} \vspace{0.02cm} $1.00$ & \vspace{0.02cm} $0.3067$ & \vspace{0.02cm} $0.4185$ & \vspace{0.02cm} $0.4349$ \\  
            \hline										
	\end{tabular}										
\end{table}

%% file: Plots/distortion_fun_J.tex
\begin{figure}[t]
\centering
\begin{tikzpicture}
\matrix (m) [
    matrix of nodes,
    row sep=3mm,
    column sep=1mm]
{
\includegraphics[width=0.462\linewidth]{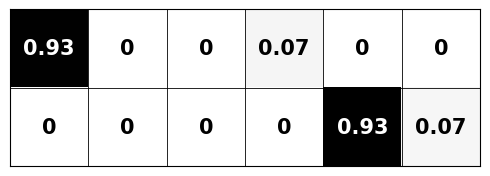} &
\includegraphics[width=0.462\linewidth]{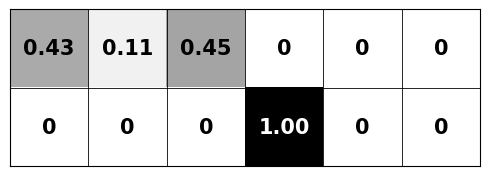} \\

\includegraphics[width=0.462\linewidth]{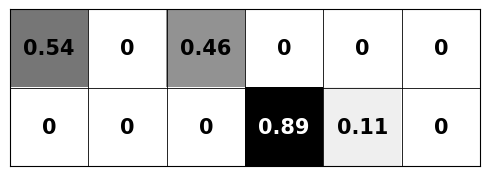} &
\includegraphics[width=0.462\linewidth]{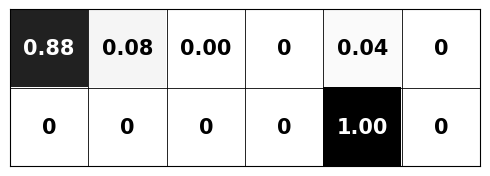} \\

\includegraphics[width=0.462\linewidth]{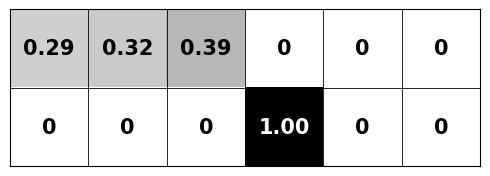} &
\includegraphics[width=0.462\linewidth]{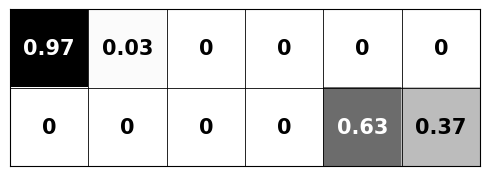} \\

\includegraphics[width=0.462\linewidth]{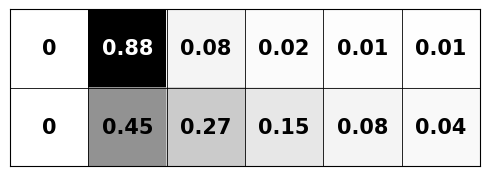} &
\includegraphics[width=0.462\linewidth]{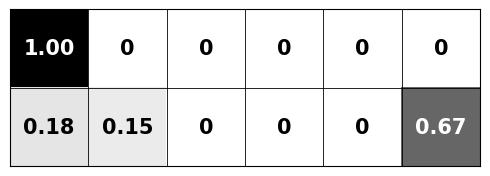} \\
};

\node at (-2.2,3.8) {\tiny{$\rho_0\qquad\;\rho_1\qquad \;\rho_2 \qquad \;\rho_3 \qquad \;\rho_4 \qquad \rho_5$}};

\node at (-2.2,1.85) {\tiny{$\rho_0\qquad\;\rho_1\qquad \;\rho_2 \qquad \;\rho_3 \qquad \;\rho_4 \qquad \rho_5$}};

\node at (-2.2, -0.15) {\tiny{$\rho_0\qquad\;\rho_1\qquad \;\rho_2 \qquad \;\rho_3 \qquad \;\rho_4 \qquad \rho_5$}};

\node at (-2.2,-2.2) {\tiny{$\rho_0\qquad\;\rho_1\qquad \;\rho_2 \qquad \;\rho_3 \qquad \;\rho_4 \qquad \rho_5$}};

\node at (2.3,-2.2) {\tiny{$\rho_0\qquad\;\rho_1\qquad \;\rho_2 \qquad \;\rho_3 \qquad \;\rho_4 \qquad \rho_5$}};

\node at (2.3, -0.15) {\tiny{$\rho_0\qquad\;\rho_1\qquad \;\rho_2 \qquad \;\rho_3 \qquad \;\rho_4 \qquad \rho_5$}};

\node at (2.3, 1.85) {\tiny{$\rho_0\qquad\;\rho_1\qquad \;\rho_2 \qquad \;\rho_3 \qquad \;\rho_4 \qquad \rho_5$}};

\node at (2.3, 3.8) {\tiny{$\rho_0\qquad\;\rho_1\qquad \;\rho_2 \qquad \;\rho_3 \qquad \;\rho_4 \qquad \rho_5$}};

\node at (0.1, 3+0.2) {\tiny{$h_0$}};
\node at (0.1, 1.3) {\tiny{$h_0$}};
\node at (0.1, -0.4-0.2) {\tiny{$h_0$}};
\node at (0.1, -2.2-0.5) {\tiny{$h_0$}};

\node at (0.1, 2.2+0.4) {\tiny{$h_1$}};
\node at (0.1, 0.5+0.1) {\tiny{$h_1$}};
\node at (0.1, -1.2) {\tiny{$h_1$}};
\node at (0.1, -3-0.4) {\tiny{$h_1$}};

\node at (-4.4, 3+0.2) {\tiny{$h_0$}};
\node at (-4.4, 1.3) {\tiny{$h_0$}};
\node at (-4.4, -0.4-0.2) {\tiny{$h_0$}};
\node at (-4.4, -2.2-0.5) {\tiny{$h_0$}};

\node at (-4.4, 2.2+0.4) {\tiny{$h_1$}};
\node at (-4.4, 0.5+0.1) {\tiny{$h_1$}};
\node at (-4.4, -1.2) {\tiny{$h_1$}};
\node at (-4.4, -3-0.4) {\tiny{$h_1$}};

\node at (2.2, 4.1) {\scriptsize{$\delta_1(\rho), \bar{D} = 0.07, \bar{\Delta} = 0.2498$}};
\node at (2.2, 2.1) {\scriptsize{$\delta_2(\rho), \bar{D} = 0.07, \bar{\Delta} = 0.7115$}};
\node at (2.2, 0.1) {\scriptsize{$\delta_3(\rho), \bar{D} = 0.07, \bar{\Delta} = 0.8531$}};
\node at (2.2, -2+0.05) {\scriptsize{$\delta_4(\rho), \bar{D} = 0.07, \bar{\Delta} = 1.2998$}};

\node at (-2, 4.1) {\scriptsize{$\delta_1(\rho), \bar{D} = 0.01,  \bar{\Delta} = 0.7771$}};
\node at (-2, 2.1) {\scriptsize{$\delta_1(\rho), \bar{D} = 0.05,  \bar{\Delta} = 0.3353$}};
\node at (-2, 0.05) {\scriptsize{$\delta_1(\rho), \bar{D} = 0.1,  \bar{\Delta} = 0.1544$}};
\node at (-2, -2+0.05) {\scriptsize{$\delta_1(\rho), \bar{D} = 0.2,  \bar{\Delta} = 0$}};

\end{tikzpicture}
\begin{tikzpicture}[scale=0.9]
\begin{axis}[
    xmin=1, xmax=5,
    ymin=-0.02, ymax=1.05,
    xlabel={$\rho$},
    ylabel={$\delta(\rho)$},
    grid=both,
    xlabel style={
        at={(axis description cs:0.5,-0.01)},
        anchor=north},
    ylabel style={
        at={(axis description cs:0.1,0.45)},
        anchor=south},
    legend style={font=\tiny, at={(0.73, 0.65)}},
    samples=400,
    domain=0:5
]

\addplot[color=red, line width=1.6pt, dashed]
    {exp(-x)};
\addlegendentry{$\delta_1(\rho)$: Convex}

\addplot[
    color=blue,
    line width=1.6pt,
    dash pattern=on 6pt off 3pt,
    domain=0:6,
    samples=400
]
{ abs(x-4) < 1e-3 ? 0.05 : (x < 4 ? 1 : 0.0) };
\addlegendentry{$\delta_2(\rho)$: Step}

\addplot[color=cyan, line width=1.6pt]
    {1 - x/5};
\addlegendentry{$\delta_3(\rho)$: Linear}

\addplot[color=black, line width=1.6pt, dotted]
    {pow(max(0, cos(deg(pi*x/10))), 0.3)};
\addlegendentry{$\delta_4(\rho)$: Concave}

\end{axis}
\end{tikzpicture}

\caption{Variation of the scheduling probabilities $\mu(h,\rho)$ (shown as table entries) and the long-term expected average VAoI under the VA-SRP,
$\bar{\Delta} = \lim_{T \to \infty} \frac{1}{T} \sum_{t=1}^{T} \mathbb{E}[\Delta(t)] = \lambda ((\sum_{\rho}\mathbb{E}_h[\mu(h, \rho)] \mathbb{I}\{\rho > 0\} )^{-1} -1)$, for different
distortion functions, $\delta(\rho)$ and bounds on the average distortion, $\bar{D}$.
The considered distortion functions are
$\delta_1(\rho)=e^{-\rho}$,
$\delta_2(\rho)=1$ for $\rho< (r_{\rm max} -1)$, $\delta_2(r_{\rm max} - 1)=0.05$, and $\delta_2(\rho)=0$ otherwise,
$\delta_3(\rho) = 1-\rho/r_{\rm max}$
and $\delta_4(\rho)=(\cos(\pi \rho/2r_{\rm \max})^{0.3}$. System parameters: $M = 1, \bar{P} = 10, \lambda = 0.9,  r_{\rm max} = 5$, $\rho \in \{\rho_0=0,\rho_1=1,\rho_2=2,\rho_3=3,\rho_4=4,\rho_5=5\}$, $\mathcal{H}=\{h_0=0.1,\,h_1=1\}$ with equal probability.
}

\label{fig:delta_fun}
\end{figure}
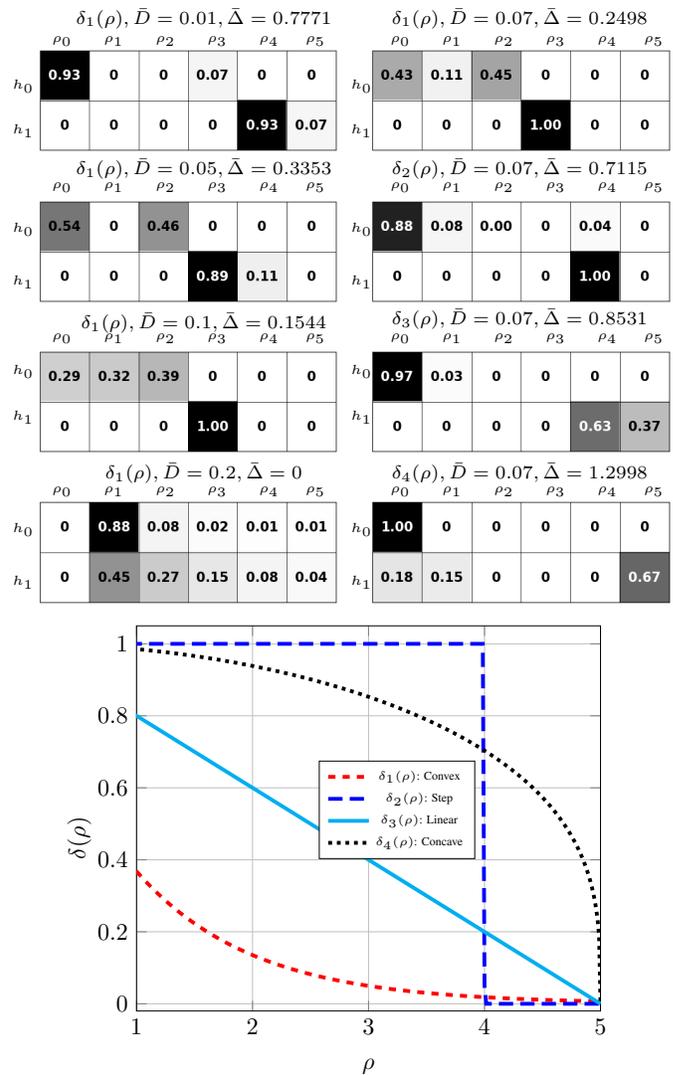

%% file: main.bbl
\begin{thebibliography}{10}
\providecommand{\url}[1]{#1}
\csname url@samestyle\endcsname
\providecommand{\newblock}{\relax}
\providecommand{\bibinfo}[2]{#2}
\providecommand{\BIBentrySTDinterwordspacing}{\spaceskip=0pt\relax}
\providecommand{\BIBentryALTinterwordstretchfactor}{4}
\providecommand{\BIBentryALTinterwordspacing}{\spaceskip=\fontdimen2\font plus
\BIBentryALTinterwordstretchfactor\fontdimen3\font minus \fontdimen4\font\relax}
\providecommand{\BIBforeignlanguage}[2]{{%
\expandafter\ifx\csname l@#1\endcsname\relax
\typeout{** WARNING: IEEEtran.bst: No hyphenation pattern has been}%
\typeout{** loaded for the language `#1'. Using the pattern for}%
\typeout{** the default language instead.}%
\else
\language=\csname l@#1\endcsname
\fi
#2}}
\providecommand{\BIBdecl}{\relax}
\BIBdecl

\bibitem{kaul2012real}
R.~D. Yates and S.~K. Kaul, ``{Status updates over unreliable multiaccess channels},'' in \emph{IEEE ISIT}, 2017, pp. 331--335.

\bibitem{pappas2023age}
N.~Pappas, M.~A. Abd-Elmagid, B.~Zhou, W.~Saad, and H.~S. Dhillon, Eds., \emph{Age of Information: Foundations and Applications}.\hskip 1em plus 0.5em minus 0.4em\relax Cambridge University Press, 2023.

\bibitem{kosta2017age}
A.~Kosta, N.~Pappas, V.~Angelakis \emph{et~al.}, ``{Age of Information: A New Concept, Metric, and Tool},'' \emph{Foundations and Trends in Networking}, vol.~12, no.~3, pp. 162--259, 2017.

\bibitem{sun2019age}
Y.~Sun, I.~Kadota, R.~Talak, and E.~Modiano, \emph{Age of Information: A New Metric for Information Freshness}.\hskip 1em plus 0.5em minus 0.4em\relax Springer Nature, 2022.

\bibitem{yates2021age}
R.~D. Yates, ``{The Age of Gossip in Networks},'' in \emph{IEEE ISIT}, 2021, pp. 2984--2989.

\bibitem{Gangadhar_VAoI}
G.~Karevvanavar, H.~Pable, O.~Patil, R.~Bhat, and N.~Pappas, ``{Version Age of Information Minimization Over Fading Broadcast Channels},'' \emph{IEEE Trans. Wireless Commun.}, vol.~24, no.~2, pp. 1620--1634, 2025.

\bibitem{baturalp2022version}
B.~Buyukates, M.~Bastopcu, and S.~Ulukus, ``{Version Age of Information in Clustered Gossip Networks},'' \emph{IEEE J. Sel. Areas Inf. Theory}, vol.~3, no.~1, pp. 85--97, 2022.

\bibitem{DelfaniTCOM25}
E.~Delfani and N.~Pappas, ``{Version Age-Optimal Cached Status Updates in a Gossiping Network With Energy Harvesting Sensor},'' \emph{IEEE Trans. Commun.}, vol.~73, no.~4, pp. 2344--2360, 2025.

\bibitem{DelfaniCOMML25}
------, ``{Semantics-Aware Updates From Remote Energy Harvesting Devices to Interconnected LEO Satellites},'' \emph{IEEE Commun. Lett.}, vol.~29, no.~8, pp. 1928--1932, 2025.

\bibitem{DelfaniINFOCOM26}
------, ``{From Timestamps to Versions: Version AoI in Single- and Multi-Hop Networks},'' \emph{arXiv preprint arXiv:2507.23433}, 2025.

\bibitem{WiOpt-21}
G.~Gagan, S.~Jayanth, and V.~B. Rajshekhar, ``{Age of Information Minimization with Power and Distortion Constraints in Multiple Access Channels},'' in \emph{IEEE WiOpt}, 2021, pp. 1--7.

\bibitem{AoI-Dist-Tradeoff4}
N.~Rajaraman, R.~Vaze, and G.~Reddy, ``{Not Just Age but Age and Quality of Information},'' \emph{J. Sel. Areas Commun.}, vol.~39, no.~5, pp. 1325--1338, 2021.

\bibitem{yates2018age}
R.~D. Yates and S.~K. Kaul, ``{The Age of Information: Real-Time Status Updating by Multiple Sources},'' \emph{IEEE Trans. Inf. Theory}, vol.~65, no.~3, pp. 1807--1827, 2018.

\bibitem{costa2016age}
M.~Costa, M.~Codreanu, and A.~Ephremides, ``{On the Age of Information in Status Update Systems With Packet Management},'' \emph{IEEE Trans. Inf. Theory}, vol.~62, no.~4, pp. 1897--1910, 2016.

\bibitem{kadota2018optimizing}
I.~Kadota, A.~Sinha, and E.~Modiano, ``{Optimizing Age of Information in Wireless Networks with Throughput Constraints},'' in \emph{IEEE INFOCOM}, 2018, pp. 1844--1852.

\bibitem{NetworkWithStochasticArrivals}
I.~Kadota and E.~Modiano, ``{Minimizing the Age of Information in Wireless Networks with Stochastic Arrivals},'' \emph{IEEE Trans. Mobile Comput.}, vol.~20, no.~3, pp. 1173--1185, 2021.

\bibitem{Bhat-JSAC}
R.~V. {Bhat}, R.~{Vaze}, and M.~{Motani}, ``{Minimization of Age of Information in Fading Multiple Access Channels},'' \emph{IEEE J. Sel. Areas Commun.}, vol.~39, no.~5, pp. 1471--1484, 2021.

\bibitem{goal_oriented_srp_2024}
R.~S. Pomaje, S.~Jayanth, R.~V. Bhat, and N.~Pappas, ``{Age of Information Minimization in Goal-Oriented Communication with Processing and Cost of Actuation Error Constraints},'' \emph{arXiv preprint arXiv:2508.07865}, 2025.

\bibitem{kumar2023}
K.~Saurav and R.~Vaze, ``{Scheduling to Minimize Age of Information With Multiple Sources},'' \emph{IEEE J. Sel. Areas Inf. Theory}, vol.~4, pp. 539--550, 2023.

\bibitem{wiopt_srp_2024}
L.~Wang, Q.~Wang, H.~H. Chen, and S.~Zhou, ``{Age of Information-Oriented Probabilistic Link Scheduling for Device-to-Device Networks},'' in \emph{IEEE WiOpt}, 2024, pp. 201--208.

\bibitem{random_access_age_2020}
X.~Chen, K.~Gatsis, H.~Hassani, and S.~S. Bidokhti, ``{Age of Information in Random Access Channels},'' \emph{IEEE Trans. Inf. Theory}, vol.~68, no.~10, pp. 6548--6568, 2022.

\bibitem{zhiguo2023age}
Z.~Ding, R.~Schober, and H.~V. Poor, ``{Age of Information: Can CR-NOMA Help?}'' \emph{IEEE Trans. Commun.}, vol.~71, pp. 6451--6467, 2023.

\bibitem{khodakhah2024balancing}
F.~Khodakhah, A.~Mahmood, {\v{C}}.~Stefanovi{\'c}, H.~Farag, P.~{\"O}sterberg, and M.~Gidlund, ``{Balancing AoI and Rate for Mission-Critical and eMBB Coexistence With Puncturing, NOMA, and RSMA in Cellular Uplink},'' \emph{IEEE Trans. Veh. Technol.}, vol.~74, no.~1, pp. 1475--1488, 2024.

\bibitem{sun2025ageinformationanalysisnomaassisted}
Y.~Sun, Y.~Ye, C.~Kai, Z.~Ding, and B.~Chen, ``{Age of Information Analysis for NOMA-Assisted Grant-Free Transmissions with Randomly Arrived Packets},'' \emph{arXiv preprint arXiv:2504.19441}, 2025.

\bibitem{9524846}
M.~Bastopcu and S.~Ulukus, ``{Age of Information for Updates with Distortion: Constant and Age-Dependent Distortion Constraints},'' \emph{IEEE/ACM Trans. Netw.}, vol.~29, no.~6, pp. 2425--2438, 2021.

\bibitem{9589182}
B.~Joshi, R.~V. Bhat, B.~Bharath, and R.~Vaze, ``{Minimization of Age of Incorrect Estimates of Autoregressive Markov Processes},'' in \emph{IEEE WiOpt}, 2021, pp. 1--8.

\bibitem{9137714}
A.~Maatouk, S.~Kriouile, M.~Assaad, and A.~Ephremides, ``{The Age of Incorrect Information: A New Performance Metric for Status Updates},'' \emph{IEEE/ACM Trans. Netw.}, vol.~28, no.~5, pp. 2215--2228, 2020.

\bibitem{10715699}
J.~Li and W.~Zhang, ``{Asymptotically Optimal Joint Sampling and Compression for Timely Status Updates: Age-Distortion Tradeoff},'' \emph{IEEE Trans. Veh. Technol.}, vol.~74, no.~2, pp. 2338--2352, 2025.

\bibitem{9611456}
Y.~İnan, R.~Inovan, and E.~Telatar, ``{Optimal Policies for Age and Distortion in a Discrete-Time Model},'' in \emph{IEEE Inf. Theory Workshop}, 2021, pp. 1--6.

\bibitem{10336741}
H.~Hong, J.~Jiao, T.~Yang, Y.~Wang, R.~Lu, and Q.~Zhang, ``{Age of Incorrect Information Minimization for Semantic-Empowered NOMA System in S-IoT},'' \emph{IEEE Trans. Wireless Commun.}, vol.~23, no.~6, pp. 6639--6652, 2024.

\bibitem{Guidan}
G.~Yao, C.-C. Wang, and N.~B. Shroff, ``{Age Minimization with Energy and Distortion Constraints},'' in \emph{Proceedings of the Twenty-fourth International Symposium on Theory, Algorithmic Foundations, and Protocol Design for Mobile Networks and Mobile Computing}, 2023, pp. 101--110.

\bibitem{gunduz_sc2019}
E.~Bourtsoulatze, D.~Burth~Kurka, and D.~Gündüz, ``{Deep Joint Source-Channel Coding for Wireless Image Transmission},'' \emph{IEEE Trans. Cognitive Commun. Netw.}, vol.~5, no.~3, pp. 567--579, 2019.

\bibitem{Sumei_Sum_BC_Region}
L.~Lei, D.~Yuan, C.~K. Ho, and S.~Sun, ``{Power and Channel Allocation for Non-Orthogonal Multiple Access in 5G Systems: Tractability and Computation},'' \emph{IEEE Trans. Wireless Commun.}, vol.~15, no.~12, pp. 8580--8594, 2016.

\bibitem{dai2015non}
L.~Dai, B.~Wang, Y.~Yuan, S.~Han, I.~Chih-lin, and Z.~Wang, ``{Non-orthogonal multiple access for 5G: solutions, challenges, opportunities, and future research trends},'' \emph{IEEE Commun. Mag.}, vol.~53, no.~9, pp. 74--81, 2015.

\bibitem{islam2017power}
S.~R. Islam, N.~Avazov, O.~A. Dobre, and K.-S. Kwak, ``{Power-Domain Non-Orthogonal Multiple Access (NOMA) in 5G Systems: Potentials and Challenges},'' \emph{IEEE Commun. Surveys Tuts.}, vol.~19, no.~2, pp. 721--742, 2016.

\bibitem{polyanskiy2010channel}
Y.~Polyanskiy, H.~V. Poor, and S.~Verdu, ``{Channel Coding Rate in the Finite Blocklength Regime},'' \emph{IEEE Trans. Inf. Theory}, vol.~56, no.~5, pp. 2307--2359, 2010.

\bibitem{Boyd}
S.~P. Boyd and L.~Vandenberghe, \emph{Convex Optimization}.\hskip 1em plus 0.5em minus 0.4em\relax Cambridge University Press, 2014.

\bibitem{Shor1985}
N.~Z. Shor, ``{The Subgradient Method},'' in \emph{{Minimization Methods for Non-Differentiable Functions}}, ser. Springer Series in Computational Mathematics.\hskip 1em plus 0.5em minus 0.4em\relax Springer, Berlin, Heidelberg, 1985, vol.~3, pp. 22--47.

\bibitem{Bertsekas2016}
D.~P. Bertsekas, \emph{{Nonlinear Programming}}, 3rd~ed.\hskip 1em plus 0.5em minus 0.4em\relax Massachusetts: Athena Scientific, 2016.

\bibitem{Nedic2009}
A.~Nedi\'{c} and A.~Ozdaglar, ``{Subgradient Methods for Saddle-Point Problems},'' \emph{Journal of Optimization Theory and Applications}, vol. 142, no.~1, pp. 205--228, 2009.

\end{thebibliography}
